\documentclass[12pt]{article}
\usepackage{graphicx}
\usepackage{amssymb,amsthm,amsfonts,amstext}
\usepackage{url}
\def\be{\begin{equation}}
\def\ee{\end{equation}}
\def\bea{\begin{eqnarray}}
\def\eea{\end{eqnarray}}
\def\bma{\begin{mathletters}}
\def\ema{\end{mathletters}}

\def\0{\overline{0}}

\def\q0{\underline{0}}

\def\H{{\cal H}}

\def\C{{\mathbb C}}
\def\id{{\mathbb I}}

\def\H{{\cal H}}

\def\B{{\cal B}}
\def\R{\mathbb{R}}
\def\N{\mathbb{N}}
\def\Z{\mathbb{Z}}

\def\tr{\mbox{tr}}
\def\dist{\mbox{dist}}
\def\M{{\cal M}}
\def\one{\leavevmode\hbox{\small1\normalsize\kern-.33em1}}

\def\bra#1{\langle#1|} \def\ket#1{|#1\rangle}
\def\braket#1#2{\langle#1|#2\rangle}

\def\proj#1{\ket{#1}\!\bra{#1}}

\newtheorem{theo}{Theorem}

\newtheorem{obs}[theo]{Observation}

\def\id{{\mathbb I}}
\def\tr{\mbox{tr}}

\begin{document}

\title{How energy conservation limits our measurements}

\author{Miguel Navascu\'es and Sandu Popescu\\
H. H. Wills Physics Laboratory, University of Bristol,\\ Tyndall Avenue, Bristol, BS8 1TL, United Kingdom}

\maketitle

\begin{abstract}
Observations in Quantum Mechanics are subject to complex restrictions arising from the principle of energy conservation. Determining such restrictions, however, has been so far an elusive task, and only partial results are known. In this paper we discuss how constraints on the energy spectrum of a measurement device translate into limitations on the measurements which we can effect on a target system with non-trivial energy operator. We provide efficient algorithms to characterize such limitations and we quantify them exactly when the target is a two-level quantum system. Our work thus identifies the boundaries between what is possible or impossible to measure, i.e., between what we can see or not, when energy conservation is at stake.
\end{abstract}

\section{Introduction}
The success of quantum computation \cite{cirac,briegel} and quantum simulation \cite{retzker,kraus} schemes depends in part on our ability to measure quantum systems with good enough precision. Sometimes (e.g., in ion-trap experiments \cite{ion}), such measurements are conducted in systems with a non-trivial energy operator, and so are strongly limited by the law of energy conservation. Indeed, Wigner was among the first to notice the impossibility of measuring exactly any observable described by an operator which does not commute with the system's conserved quantities. He showed, nevertheless, that an arbitrarily close measurement of such observables was possible if the dimensions of the measurement device were large enough \cite{wigner}. These results were formalized in the Wigner-Araki-Yanase (WAY) theorem \cite{araki,yanase} and quantified some time later by Ozawa \cite{ozawa}, who provided a general uncertainty relation which bounds the mean square noise in the measurement of an arbitrary observable as a function of the variances of the conserved quantities in system and measurement apparatus. This relation has been applied successfully to estimate the error of certain lab-induced evolutions of two-level quantum systems \cite{gate_fid,gate_fid2}. Sadly enough, and despite its generality and range of applicability, Ozawa's uncertainty relation is not tight, and sometimes greatly underestimates the errors it tries to bound. 

In this article we study how measurements of a target system effected by a quantum device are limited by the energy distribution of the latter. We quantify such limits analytically when the target is a two-level system. Additionally, we provide an efficient algorithm to completely characterize the set of attainable measurements in arbitrarily high dimensional targets.

The structure of the article is as follows: in Section \ref{bell} we will motivate our study by analyzing an experimental scheme for Bell inequality violation by means of energy-conserving transformations. Then, in Section \ref{considerations}, we will describe our measurement model and explain the role of the quantum device's energy spectrum in the measuring process. To quantify the effect of the energy spectrum on the set of available measurements, we introduce in Section \ref{quantifying} two operational distances between an arbitrary pair of quantum measurements. Using these notions, in Section \ref{two_levels} we find that the difference between the set of all conceivable two-level measurements and the set of all measurements implementable by measurement devices with a battery in one of the states $\B=\{\sigma_B\}$ can be quantified operationally by a single parameter, $\tau(\B)$. We will calculate this parameter in two interesting scenarios: 1) measurement devices with finite energy spectrum; and 2) measurement devices with unbounded energy spectrum, but finite average energy $\bar{E}$. Later, in Section \ref{charact}, we will describe an efficient algorithm to characterize the set of accessible measurements for target systems of any dimension. Finally, in Section \ref{paradox} we will propose a physical mechanism to explain why we should expect to estimate non-trivial observables when the spectra of target system and measurement apparatus are just approximately resonant.

Before starting, though, we would like to call attention to the recent and related work of Ahmadi et al. \cite{mehdi}, which also analyzes how the quantum state of the control device influences its capabilities for quantum measurements and control. Among its main results the reader can find a reformulation of the WAY theorem in the language of resource theories \cite{resource1,resource2} and the realization that, in finite dimensional systems, there does not exist such a thing as an optimal universal ancillary state, i.e., the optimal quantum state of the control device will depend on which specific quantum operation we wish to implement in the lab.

\section{A Bell experiment under energy conservation}
\label{bell}

Let two parties, call them Alice and Bob, be distributed the entangled state $\ket{\varphi}_{AB}\equiv\frac{1}{\sqrt{2}}(\ket{0}_{A}\ket{1}_{B}+\ket{1}_{A}\ket{0}_B)$, expressed in the Fock basis (where $\ket{n}$ represents a state with $n$ photons). We will show that, when restricted to apply (energy-conserving) passive optical interactions, Alice and Bob cannot violate any Bell inequality. However, they can violate the Clauser-Horne-Shimony-Holt (CHSH) Bell inequality \cite{chsh} if, in addition, they are given each a copy of the state $\ket{+}=\frac{1}{\sqrt{2}}(\ket{0}+\ket{1})$.

\begin{figure}
  \centering
  \includegraphics[width=9 cm]{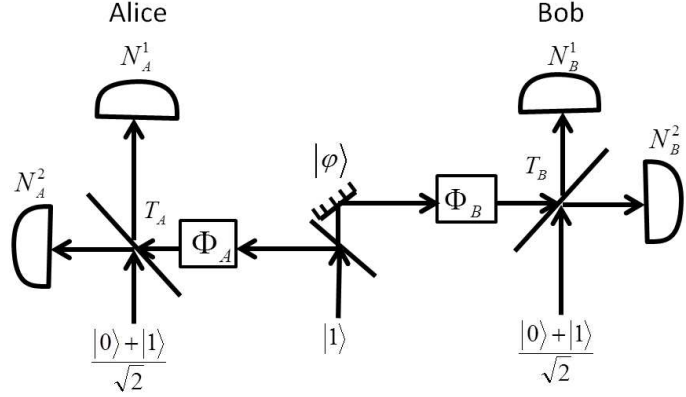}
    \caption{{\bf Bell experiment under energy-conserving interactions.} Alice and Bob can violate the CHSH Bell inequality by means of beam splitters of varying transmitivity $T_A,T_B$ and shifters of phase $\Phi_A,\Phi_B$ if, besides the entangled state $\ket{\varphi}$, they are distributed each a copy of the reference state $\ket{+}$. The quantity $N_A^k$ ($N_B^k$) denotes the number of photons registered by Alice's (Bob's) $k^{th}$ detector.}
 \label{Bell}
\end{figure}

First, as beam-splitting and phase-shifting transformations (and photo-counting) commute with the total photon number, if Alice and Bob are given the state $\ket{\varphi}_{AB}$, the statistics they will observe are indistinguishable from those generated by the locally dephased state $\frac{1}{\sqrt{2}}(\proj{0}_A\otimes\proj{1}_B+\proj{1}_A\otimes\proj{0}_B)$. Since such is a separable state, under no circumstances they will achieve a Bell inequality violation.

Suppose now that Alice and Bob are also distributed each a copy of the state $\ket{+}$. After local dephasing, their joint state $\ket{\varphi}_{AB}\ket{+}_{A'}\ket{+}_{B'}$ turns into

\begin{eqnarray}
&\rho_{AA'BB'}=\frac{1}{8}\left\{\proj{0,0}_{AA'}\otimes (\proj{1,0}_{BB'}+\proj{1,1}_{BB'})+\right.\nonumber\\
&+(\proj{1,0}_{AA'}+\proj{1,1}_{AA'})\otimes \proj{0,0}_{BB'}+\nonumber\\
&+\proj{1,1}_{AA'}\otimes\proj{0,0}_{BB'}+\proj{0,0}_{AA'}\otimes\proj{1,1}_{BB'}+\nonumber\\
&\left. +2\proj{\psi}_{AA'BB'}\right\},
\label{dephased}
\end{eqnarray}

\noindent with $\ket{\psi}=\frac{1}{\sqrt{2}}(\ket{0,1}_{AA'}\ket{1,0}_{BB'}+\ket{1,0}_{AA'}\ket{0,1}_{BB'})$.


Since with probability $\frac{6}{8}$ the density matrix $\rho_{AA'BB'}$ is a separable state; and with probability $\frac{2}{8}$, a maximally entangled state, no matter which measurements Alice and Bob apply, they will never be able to violate CHSH by more than $\frac{6}{8}\times 2+\frac{2}{8}\times 2\sqrt{2}\approx 2.2071$

Now, define $\sigma_x\equiv\ket{0,1}\bra{1,0}+\ket{1,0}\bra{0,1}$, $\sigma_z\equiv\proj{0,1}-\proj{1,0}$, and, for $k=1,2$, consider the following dichotomic operators:

\begin{eqnarray}
&A_k=\proj{0,0}-\proj{1,1}+\tilde{A}_k,\nonumber\\
&B_k=\proj{1,1}-\proj{0,0}+\tilde{B}_k,
\end{eqnarray}

\noindent with

\be
\tilde{A}_1=\sigma_z,\tilde{A}_2=\sigma_x,\tilde{B}_1=\frac{1}{\sqrt{2}}(\sigma_x-\sigma_z),\tilde{B}_2=-\frac{1}{\sqrt{2}}(\sigma_x+\sigma_z).
\ee

\noindent It can be verified that

\begin{eqnarray}
&\tr\left\{(\rho_{AA'BB'} (A_1\otimes B_1+ A_1\otimes B_2+ A_2\otimes B_1- A_2\otimes B_2)\right\}=\nonumber\\
&=1+\frac{3}{4}\sqrt{2}\approx 2.0607,
\end{eqnarray}

\noindent i.e., when these measurements are applied on state $\rho_{AA'BB'}$, they generate a CHSH value beyond the classical limit.

Given the setup depicted in Figure \ref{Bell}, Alice and Bob can implement measurements $A_1,A_2,B_1,B_2$ by fixing appropriately the values of their transmitivities and phase shifts, and assigning labels $\pm1$ to the possible outcomes $N_A^1,N_A^2, N_B^1,N_B^2$ of their photo-detectors via the prescriptions:

\begin{eqnarray}
A=&1,& \mbox{ if } N_A^1+N_A^2=0, \mbox{ or } N_A^1=0,N_A^2=1\nonumber\\
&-1,& \mbox{ if } N_A^1+N_A^2=2, \mbox{ or } N_A^1=1, N_A^2=0;
\end{eqnarray}

\begin{eqnarray}
B=&1,& \mbox{ if } N_B^1+N_B^2=2, \mbox{ or } N_B^1=0,N_B^2=1\nonumber\\
&-1,& \mbox{ if } N_B^1+N_B^2=0, \mbox{ or } N_B^1=1, N_B^2=0.
\end{eqnarray}

\noindent It follows that the presence of coherent superpositions of energy eigenstates $\ket{+}$ increases the set of local measurements which Alice and Bob can implement over the state $\ket{\varphi}$. 

It can be seen that, when distributed highly energetic coherent states rather than $\ket{+}$ as an ancilla, Alice and Bob can measure any observable in the subspace $\mbox{span}(\ket{0},\ket{1})$ up to arbitrary precision, and so they can get as close as they want to the maximal quantum violation of CHSH. Our work focuses precisely on how the energy spectrum of such ancillary states constrains the set of achievable measurements in quantum systems with non-trivial energy operator.

\section{Measurements under energy conservation}
\label{considerations}

\subsection{The measurement model}

A general measurement over a quantum system in state $\rho\in B(\H)$ is described by a set of positive semidefinite operators $\{M_x\}_{x=0}^n$, with $\sum_{x=0}^n M_x=\id_\H$, such that 

\be
p(x)=\tr(\rho M_x)
\label{probabi}
\ee

\noindent denotes the probability of obtaining the outcome $x$. In general, such a measurement is physically realized by attaching our \emph{target} system $S$ to a \emph{pointer} $P$, and making them interact via a third system $C$ that we will denote the \emph{clock}. 


We will assume that the clock does not transfer energy to the system under observation, i.e., if we call $\Omega$ the quantum channel induced by the clock on to systems $S,P$, then $\Omega(\rho_{SP})$ must have the same energy distribution as $\rho_{SP}$ for whatever initial state $\rho_{SP}$ of the target and the pointer. 

Such is the case, for instance, when the clock is a macroscopic object and the interaction is introduced adiabatically. Indeed, let $g(t)H^{\mbox{int}}_{SP}$ be the interaction effected by the clock on system $SP$ at time $t$, with $g(-\infty)=g(\infty)=0$. If $g(t)$ varies slowly enough, the adiabatic theorem guarantees (modulo level-crossings) that the energy of system $SP$ will be the same at the end of the process \cite{adiabatic}. In this measurement model, the outcome $x$ of the effective measurement on system $S$ is obtained by conducting a von Neumann measurement of the state of the pointer.

As we will see, energy conserving measurements are very limited, so we must include in the picture a fourth quantum system, the \emph{battery} ($B$), to model the energy exchange between the target and the measurement apparatus. The battery has a non-trivial energy operator $H_B$, and only interacts with the system $SP$ via the clock.

Some clarifications about this model are in order: we divide the measurement device into three different parts (the clock, the pointer and the battery) to separate different effects so that we can easily follow and analyze them. A clock, by definition, must change its state, hence, first of all, its Hamiltonian cannot be identically zero and second, it cannot be in an energy eigenstate, but it must be in a superposition of different energy eigenstates. In general a clock may exchange energy with the system - turning on and off an interaction may excite the system. So when we discuss the energy exchange with the system we must take the clock into account as well. This significantly complicates the analysis. We thus prefer to separate the issue of energy exchange from the clock. When the turning on and off is adiabatic, such an energy exchange is prevented.  The entire energy exchange is then with the battery. Effectively this way we succeed to separate the clock from the energy exchange, as intended. We also want to be able to keep track of all possible entanglements and to use the clock only for timing. If the target is to get entangled with some other system, we provide this as an explicit ancilla (that we can associate with the battery, for instance). Hence we take the clock as `classical'. The main point of all this is that such a clock allows us to produce a well defined energy conserving unitary. Note, finally, that this formalism actually already contains implicitly a non-adiabatic, non-classical clock, if we regard part of the battery as part of the clock. Hence our simplified model loses no generality. 

The interaction $U$ between the target, the pointer and the battery mediated by the clock will thus be such that the distribution of the total energy $H_T$ of the composed system $SPB$ will be conserved. In other words, if $H_T=\sum_n E_nP_n$ is the spectral decomposition of the total hamiltonian, then $\tr\{\rho_{SPB}P_n\}=\tr\{U\rho_{SPB} U^\dagger P_n\}$ must hold for all joint states $\rho_{SPB}$ of target, pointer and battery, and all $n$. This condition can be seen equivalent to $[U,H_T]=0$, see Appendix \ref{general}. In Appendix \ref{general} it is also shown that any completely positive trace-preserving map $\Omega$ that conserves the energy distribution can be understood as a unitary $U$ acting on $SPB$ and an ancillary system $A$ such that $\Omega(\rho_{SPB})=\tr_A\{U(\rho_{SPB}\otimes \proj{0}_A)U^\dagger\}$ and $[U,H_T\otimes \id_A]=0$. Thus even if we relaxed our previous assumptions and allowed the clock to get entangled with system $SPB$, we could always redefine the battery as $BA$ (with hamiltonian $H_B\otimes \id_A$) and again view the clock's interaction as a unitary commuting with $H_T$.

What is the total energy $H_T$ in this model? Note that the pointer can always be chosen such that its local energy operator is trivial\footnote{Given that the pointer must hold a reliable record of the outcome of the measurement, the states $\{\ket{a}\}$ must be orthogonal and do not evolve in time, i.e., they must be eigenstates of the hamiltonian $H_P$ of the pointer. If the energy of such states is equal, then we can assume that $H_P=0$. If not, we can always attach a pointer $P'$ to $P$ with $H_{P'}=0$ and make the joint system evolve via the unitary $V=\sum_{a}\proj{a}_P\otimes D(a)_{P'}$, where $D(a)\ket{k}=\ket{k+a \mbox{ }(\mbox{mod } n+1)}$. If $P'$ is initially set to the state $\ket{0}$, this unitary will copy the measurement information to $P'$. Moreover, $[H_P\otimes \id_{P'},V]=0$, i.e., $V$ represents an energy-conserving operation. We can hence take $P'$ to be our actual pointer and regard $P$ as part of the battery system.}. Denoting by $H_S$ the energy operator of system $S$, this means that

\be
H_T=H_S\otimes \id_{PB}+\id_{SP}\otimes H_B.
\ee

For illustration, think of homodyne measurements in quantum optics: there, the target is the laser beam to be measured, and the battery is a high energy light pulse derived from a local oscillator. The displacement of these two beams through optical fibers plays the role of the clock, which switches on and off an interaction with a beam splitter and two photodetectors. Finally, the difference between the intensity of both photodetectors is then printed on a piece of paper (the pointer).

\subsection{The role of the battery}

Suppose that we wish to measure our target system by means of a battery-less quantum device. Appendix \ref{general} shows that the POVM elements $\{M_x\}$ describing any such measurement must commute with $H_S$. If, in addition, system $S$ has a non-degenerate energy operator $H_S=\sum_m E_m\proj{m}$, then each $M_x$ must admit a diagonalization of the form $M_x=\sum_m p_x^{(m)}\proj{m}$. By the completeness relation, we further have that $\sum_x p_x^{(m)}=1$ for all $m$. It follows that the measurement of any property $x$ can be simulated by the following process:

\begin{enumerate}
\item Measure the energy operator $H_S$, thus obtaining an outcome $m$ (corresponding to the energy value $E_m$).
\item Generate $x$ randomly according to the distribution $p(x)=p_x^{(m)}$.
\end{enumerate}

Measurements in this scenario are thus not very `quantum', in the sense that all we need to return an outcome is the expression of the energy density of system $S$. In particular, no matter how much entangled system $S$ is with another system $S'$, two experimentalists at each site could never violate a Bell inequality or even prove that their joint state is entangled. Indeed, let $\rho_{SS'}$ be the joint state of systems $S,S'$. Then, if the experimentalist at $S$ is restricted to perform measurements which commute with $H_S$, it is easy to see that $\rho_{SS'}$ will produce the same bipartite measurement statistics as the classical-quantum state

\be
\tilde{\rho}_{SS'}=\sum_{m}p(m)\proj{m}_S\otimes \rho^{(m)}_{S'}.
\ee

\noindent where $p(m)$ is the probability that system $m$ has energy $E_m$ and $\rho^{(m)}=\bra{m}\rho_{SS'}\ket{m}/p(m)$.

How does the picture change when our measurement device has a battery? Let us assume that such is the case, and that $H_B=\sum_n \mu_n \proj{n}$ is non-degenerate. This last condition is not restrictive in the least, since we can always introduce degeneracy later via the pointer.

As before, any effective measurement in the joint system $SB$ will be described by a complete set of POVM elements $\{M_x\}$ commuting with the energy operator of the system. The difference stems in this case in that the new energy operator is not $H_S$, but $H_S\otimes \id_{B}+\id_{S}\otimes H_B$. 

This can make a huge difference. Coming back to the example of homodyne measurements, note that a measurement of the quadrature $x$ cannot be performed via passive linear optical elements and photodetectors alone. Indeed, unlike in homodyne measurements, any such measurement would commute with the total photon number. Homodyne measurements can nevertheless be performed via linear optical elements up to arbitrarily good approximations if we introduce a second high energy laser pulse (i.e., a battery).

Note that

\be
H_T=\sum_{m,n} (E_m+\mu_n)\proj{m}\otimes \proj{n}.
\label{total_energy}
\ee

\noindent We will say that there exists a resonance between the hamiltonians $H_B,H_S$ iff there exist $m,m',n,n'$ such that

\be
E_m-E_{m'}= \mu_n-\mu_{n'}.
\ee

In the particular case that there are no resonances, the eigenspaces of $H_T$ are given by $\{\ket{m}\otimes \ket{n}\}$. Any operator $M$ commuting with $H_T$ will necessarily be of the form $M=\sum_{m,n} M_{n,m}\proj{m}\otimes \proj{n}$. Let $\sigma_B\in B(\H_B)$ be the initial state of the battery. Then, any measurement $\{M_x\}$ performed over the system $SB$ through this scheme can be simulated by the process:

\begin{enumerate}
\item 
Measure the energy of system $S$, obtaining outcome $m$.
\item
Output $x$ randomly from the distribution $p(x)=\tr(\bra{m}M_x\ket{m}\sigma_B)$.

\end{enumerate}

\noindent Non-resonant hamiltonians $H_B$ thus do not provide any advantage towards measuring or interacting with system $S$ in a quantum way.

What about resonant hamiltonians? Suppose that our target system $S$ is a two-level system, with $H_S=\proj{1}$, where $\ket{1}$ ($\ket{0}$) denotes the excited (ground state) of $S$. Let us study which kind of \emph{effective} measurements we can implement on system $S$ if we couple it to a battery with hamiltonian $H_B=\sum_{k=0}^{d-1}k\proj{k}$.

The total energy of the system can be seen equal to

\begin{eqnarray}
&H_T=\sum_{k=1}^{d-1} k\{\proj{0}\otimes\proj{k}+\proj{1}\otimes\proj{k-1}\}+\nonumber\\
&+d\proj{1}\otimes \proj{d-1}.
\end{eqnarray}

\noindent The eigenspaces of $H_T$ are thus defined by the projectors

\begin{eqnarray}
&\proj{0}\otimes \proj{0}, \nonumber\\
&\{\proj{0}\otimes\proj{k}+\proj{1}\otimes\proj{k-1}\}_{k=1}^{d-1},\nonumber\\
&\proj{1}\otimes \proj{d-1}\}_j.
\end{eqnarray}

Call the corresponding spaces $\H^{k}$, with $\H^{0}=\mbox{span}\{\ket{0}\ket{0}\}$ and $\H^{d}=\mbox{span}\{\ket{1}\ket{d-1}\}$. Then, a generic POVM element in $\H^{k}$ has the form 

\be
M^{k}=\sum_{a,b=0,1}(M^{k})_{ab}\ket{a}\bra{b}\otimes\ket{k-a}\bra{k-b},
\label{physical_ex}
\ee

\noindent where $\ket{-1}=\ket{d}=0$ by definition. We will call such POVMs \emph{physical}, since they correspond to the actual operation effected on the quantum system $SB$.

Let $\sigma_B$ be the state of the battery. Then, a physical POVM of the form $M_x=\oplus_{k} M^{k}_x$ induces in system $S$ an \emph{effective} POVM

\be
(\tilde{M}_x)_{ab}=\sum_{k=0}^{d}\bra{k-b}\sigma_B\ket{k-a}(M^{k}_x)_{ab}.
\label{many_chains_mix_ex}
\ee

Now, take $\sigma_B=\proj{\psi_B}$, with $\ket{\psi_B}=\frac{1}{\sqrt{d}}\sum_{k=0}^{d-1}\ket{k}$ and let $\{M_x\}_x$ be an arbitrary two-level POVM that we want to approximate. Set the physical POVM elements to be equal to $\{M_x\}_x$, i.e., $(M^k_x)_{ab}=(M_x)_{ab}$. Then, one can verify that

\be
(\tilde{M}_x)_{ab}=(M_x)_{ab}\left(1+\frac{\delta_{ab}-1}{d}\right).
\ee

\noindent It thus follows that any measurement can be approximated up to arbitrary precision by taking $d$ large enough. It is straightforward to extend this result to multi-level target systems.

Notice that, in order to perform non-trivial quantum measurements over the previous system, an \emph{exact} resonance between $H_B$ and $H_S$ is required. This is certainly counter-intuitive: one would expect that ancillary systems with energy operator $\tilde{H}_B\approx H_B$ \emph{nearly} resonant with $H_S$ should induce similar effective measurements over system $S$ (and thus approximate the set of all possible two-level measurements for $d\gg 1$). In Section \ref{paradox} we provide a possible solution to this apparent paradox, by invoking the existence of hidden continuous degrees of freedom.

The aim of the rest of the article is to determine exactly how the cardinality of the spectra and/or the energy of our measurement device constrain the set of effective POVMs that such a device is able to implement on its target system. But, before this, we will have to specify means to quantify such constraints.

\section{Quantifying the size of the set of accessible measurements}
\label{quantifying}

Imagine that we hold a measurement device whose battery we can initially prepare in a set of states ${\cal B}=\{\sigma_B\}$\footnote{Note that in the main text we identified ${\cal B}$ with a set of energy distributions, and considered all possible battery states with energy density in ${\cal B}$. Since that is a particular set of quantum battery states, the statements in the main article follow from the more general claim stated here.}, and call $\M(\B,d)$ the set of all effective POVMs which it allows to implement in a target system of dimension $d$ and energy operator $H_S\in B(\C^d)$. In the following section, we will try to quantify how $\M(\B,d)$ compares with $\M(d)$, the set of all possible POVMs in $\C^d$.

To do so, we must first introduce a natural distance between two different POVMs. Suppose that we have a device capable of implementing either the measurement $M^0\equiv\{M^0_x\}_x$ or $M^1\equiv\{M^1_x\}_x$ with probability $1/2$. We let it act over a suitably prepared quantum state $\rho$ and, from the outcome $x$ obtained, we try to guess which of the two POVMs our machine is actually implementing. It can be shown (see Appendix \ref{distances}) that the maximum probability $P_C$ of correctly guessing the POVM is then given by

\be
P_C=\frac{1}{2}\{1+\dist_C(M^0,M^1)\},
\ee

\noindent where $\dist_C(M^0,M^1)$ is defined as

\be
\dist_C(M^0,M^1)=\frac{1}{2}\max_{\rho}\sum_x |\tr\{\rho(M^0_x-M^1_x)\}|,
\label{dist_C}
\ee

\noindent and the maximization is to be performed over all normalized quantum states $\rho$. We will call this latter expression the \emph{classical distance} between POVMs $M^0$, $M^1$. It satisfies the triangle inequality (i.e., it is a proper distance), and has maximum value $1$. The reason why we call it classical is that, in the previous protocol, the POVM is guessed by analyzing the classical data $x$. The classical distance is somehow related to the distance between quantum maps introduced in \cite{geza}.

Analogously, we can define a \emph{quantum distance} between POVMs, by means of a protocol where the player is allowed to input part of an entangled state $\rho_{DQ}$ in the measurement device, which then performs a demolition measurement on system $D$. Depending on the outcome $x$ of such a measurement, the player will implement a POVM $N_{a}^x$, with outcomes $a\in\{0,1\}$ on system $Q$ in order to decide which of the two POVMs, $M^0$ or $M^1$, is actually being measured. The probability $P_Q$ of correctly guessing the POVM can then be seen equal to

\be
P_Q=\frac{1}{2}\{1+\dist_{Q}(M^0,M^1)\},
\ee

\noindent with

\be
\dist_Q(M^0,M^1)=\frac{1}{2}\max_{\rho_{DQ}}\sum_x\|\tr_D(\rho_{DQ} (M^0_x-M^1_x)\otimes \id_Q\|_1
\label{dist_Q}
\ee

\noindent where $\rho_{DQ}$ varies over all possible states $\rho_{DQ}\in B(\H_D\otimes \H_Q)$ and all possible Hilbert spaces $\H_Q$. See Appendix \ref{distances} for a proof. As the classical distance, $\dist_Q$ satisfies the triangle inequality and has maximum value $1$. Also, $\dist_Q(M^0,M^1)\geq \dist_C(M^0,M^1)$, for all $M^0$, $M^1$. $\dist_Q(M^0,M^1)$ actually corresponds to the diamond norm \cite{dorit} between the quantum channels $\Omega^a(\rho)=\sum_x\tr(\rho M_x)\proj{x}$, with $a=1,2$.

It can be seen that, if $M^0,M^1$ are two-outcome POVMs, $\dist_Q(M^0,M^1)= \dist_C(M^0,M^1)$. However, even in dimension 2 there exist examples of POVMs where $\dist_Q(M^0,M^1)> \dist_C(M^0,M^1)$ (Appendix \ref{distances}).

These two distances suggest a simple way to quantify the difference between a particular set of POVMs $\M'$ acting on $\C^d$ and the set $\M(d)$ of all possible quantum measurements in that space, by computing the maximum distance between an arbitrary element of $\M(d)$ and the set $\M'$. This intuition leads to the following definitions:

\be
\epsilon_C(\M')=\max_{M\in\M(d)}\dist_C(M,\M');\epsilon_Q(\M')=\max_{M\in\M(d)}\dist_Q(M,\M').
\ee

\noindent $\epsilon_C(\M')$, $\epsilon_Q(\M')$ thus measure the worst-case probability of correctly distinguishing a general POVM $M\in\M(d)$ from its closest element in $\M'$ in classical and quantum protocols, respectively. Intuitively, $\epsilon_C(\M')$, $\epsilon_Q(\M')$ measure the ability to distinguish between a device capable of implementing any measurement in $\M(d)$ and another one restricted to POVMs in $\M'$.

We have just defined two quantities to measure the performance of a quantum measurement device. Our next step will be to evaluate such quantifiers when the target system is a two-level quantum system.

\section{Two-level systems}
\label{two_levels}

\subsection{General considerations}
\label{consid}

As in section \ref{considerations}, suppose that our target system $S$ is a two-level system with $H_S=\Delta\proj{1}$, which we couple to a battery $B$ with hamiltonian $H_B$.

For the time being, assume that the spectrum of $H_B$ is discrete. We say that an increasing sequence of eigenvalues of $H_B$ forms a chain of length $L$ iff it can be written as $(\nu+\Delta k)_{k=0}^{L-1}$, for some $\nu\in \R$. We call a chain maximal if it is not a subset of a larger chain. It is clear that each $\mu\in \mbox{spec}(H_B)$ belongs to a unique maximal chain. Hence we have that

\begin{eqnarray}
&H_T=\sum_j\nu_j\proj{0}\otimes \proj{j,0} +\nonumber\\
&+ \sum_{k=1}^{L(j)-1} (\nu_j+k\Delta)\{\proj{0}\otimes\proj{j,k}+\proj{1}\otimes\proj{j,k-1}\}+\nonumber\\
&+(\nu_j+ L(j)\Delta)\proj{1}\otimes \proj{j,L(j)-1},
\end{eqnarray}

\noindent where each $j$ denotes a maximal chain of length $L(j)$ and $\ket{j,k}$ is the normalized eigenvector of $H_B$ with eigenvalue $\nu_j+k\Delta$.

The eigenspaces of $H_T$ are thus

\begin{eqnarray}
&\proj{0}\otimes \proj{j,0}, \nonumber\\
&\{\proj{0}\otimes\proj{j,k}+\proj{1}\otimes\proj{j,k-1}\}_{k=1}^{L(j)-1},\nonumber\\
&\proj{1}\otimes \proj{j,L(j)-1}\}_j.
\end{eqnarray}

In analogy with the previous section, call these spaces $\H^{j,k}$, with $\H^{j,0}=\mbox{span}(\ket{0}\ket{j,0})$ and $\H^{j,L(j)}=\mbox{span}(\ket{1}\ket{j,L(j)-1})$, and let $\sigma_B$ be the state of the battery. Then, a physical POVM $\{M_x=\oplus_{k,j} M^{k,j}_x\}$, with 

\be
M^{j,k}_x=\sum_{a,b=0,1}(M^{j,k}_x)_{ab}\ket{a}\bra{b}\otimes\ket{j,k-a}\bra{j,k-b},
\label{physical}
\ee

\noindent induces in system $S$ an \emph{effective} POVM

\be
(\tilde{M}_x)_{ab}=\sum_{j}\sum_{k=0}^{L(j)}\tr\{\sigma_B\ket{j,k-a}\bra{j,k-b}\}(M^{j,k}_x)_{ab}.
\label{many_chains_mix}
\ee

If $H_B$ is of the form (\ref{hamil_A}), by eq. (\ref{many_chains_mix}), a measurement device possessing a battery $B'$ with energy operator $H_{B'}=\sum_{n=0}^{L-1} n\proj{n}$ (where $L=\max_j L(j)$) can simulate the above measurement via the following protocol:

\begin{enumerate}
\item
Choose $j$ randomly according to the distribution $\{p_j\}$, with $p_j\equiv \sum_{k=0}^{L(j)-1}\bra{j,k}\sigma_B\ket{j,k}$.
\item
Prepare the state $\sigma_j\propto \sum_{k,k'=0}^{L(j)-1} \bra{j,k}\sigma_B\ket{j,k'}\ket{k}\bra{k'}$ in the ancillary system $B'$.
\item
Measure the joint system $SB'$ with the physical POVM $\{M^j_x=\oplus_k M^{j,k}_x\}$ (after any arbitrary completion, note that $\sum M^j_x\leq\id$).
\end{enumerate}

\subsection{How close is our measuring device from being perfect?}
\label{perfect}
Imagine that our preparation devices are capable of setting up any battery state $\rho_B\in\B$ at the beginning of the experiment. Suppose further that the energy operator of system $B$ is of the form

\be
H_B=\sum_{j}\sum_{k=0}^{L(j)-1}(\nu_j+ k\Delta)\proj{k,j}
\label{hamil_A}.
\ee

In Appendix \ref{role_tau} we show that, given those conditions,

\be
\epsilon_C(\M_{\B})=\epsilon_Q(\M_{\B})=\frac{1}{2}\{1-\tau(\B)\},
\ee

\noindent where

\be
\tau(\B)=\max_{\sigma\in \B}\sum_j\sum_{k=0}^{L(j)-2}\left|\bra{j,k+1}\sigma\ket{j,k}\right|.
\label{def_tau}
\ee

Moreover, let $\sigma^\star\in\B$ be any state maximizing (\ref{def_tau}). Then, for any general one-qubit POVM $M=\{M_x\}_x\in \M$, $\sigma^\star$ allows to generate an effective two-level POVM $\hat{M}=\{\hat{M}_x\}$, with 

\be
(\hat{M}_x)_{ab}=(M_x)_{ab}\{[1-\tau(\B)]\delta_{ab}+\tau(\B)\}.
\ee

\noindent The resolution of (\ref{def_tau}) thus allows to define a prepare-and measure strategy which simulates the whole set $\M$ with accuracy $1-\tau(\B)$.

Correspondingly, if the ancillary system has a continuous energy spectrum, the value of $\tau(\{\sigma_B\})$ equals:

\be
\tau(\{\sigma_B\})=\int dE \bra{E}\sigma_B\ket{E+\Delta}.
\ee

\noindent In the particular case that $\sigma_B$ is a pure state with energy density $f(E)dE$, we can re-express this last equation as

\be
\tau(\{\sigma_B\})=\int dE f^{1/2}(E)f^{1/2}(E+\Delta).
\label{continuous}
\ee

\noindent For example, take $f(E)dE\approx\frac{e^{\frac{-E^2}{2\sigma}}}{\sqrt{2\pi\sigma}}dE$. Then, $\tau(\{\sigma_B\},\Delta)\approx e^{-\frac{\Delta^2}{8\sigma}}$, i.e., the measurement device state is useless to measure systems with $\Delta\gg \sqrt{\sigma}$.

\subsection{Finite spectrum}
\label{finite_spec}

Picture an experimental scenario where we only have control over the first $d$ energy levels of our battery system. That is, we can only prepare battery states with energy density $f(E)=\sum_{k=0}^{d-1}p_k\delta(E-E_k)$, for some energy values $\{E_k\}$. Equivalently, suppose that the spectrum of $H_B$ is finite, i.e., that any measurement of the energy of the battery can only produce a finite set of outcomes. Denote by $d$ the cardinality of $\mbox{spec}(H_B)$. What is the set of measurements ${\cal M}(\C^d,2)$ which one can perform over system $S$ with this model, and how does it differ from the set of all possible POVMs ${\cal M}(2)$? Most importantly, how fast does ${\cal M}(\C^d,2)$ tend to ${\cal M}(2)$ in the limit $d\to \infty$?

As we saw in section \ref{consid}, the existence of more than one chain in the spectrum of the battery does not provide any advantage. In the following we will hence assume that $H_B=\sum_{k=0}^{d-1} k\Delta\proj{k}$.

Our aim is therefore to study the set of effective POVMs implementable via measurement apparatuses equipped with a battery of energy distribution $\{p_k\geq 0:\sum_{k=0}^{d-1}p_k=1\}$. By Appendix \ref{algorithms}, we can take the battery states to be pure, i.e., $\ket{\psi}=\sum_{k=0}^{d-1}c_k\ket{k}$. Then, according to the last section, in order to quantify the difference between $\M(\C^d)$ and $\M$, all we have to do is compute

\be
\tau(\C^d)=\max\left\{\sum_{k=0}^{d-2}|c_kc_{k+1}|:\sum_{k=0}^{d-1}|c_{k}|^2=1\right\}.
\label{tau_CD}
\ee

In Appendix \ref{eigen_finite} we show that the solution of (\ref{tau_CD}) is $\cos\left(\frac{\pi}{d+1}\right)$. It follows that

\be
\epsilon_C(\M(\C^d))=\epsilon_Q(\M(\C^d))=\frac{1}{2}\left\{1-\cos\left(\frac{\pi}{d+1}\right) \right\}.
\ee

Curiously enough, the states $\ket{\psi^\star_d}$ maximizing eq. (\ref{def_tau}) happen to have a non-trivial energy distribution, given by 

\be
p(k)=\frac{2}{d+1}\sin^2\left(\frac{(k+1)\pi}{d+1}\right), k=0,...,d-1.
\label{distrib_star}
\ee

\subsection{Finite energy}

Consider now a scenario where in principle we can prepare any initial battery state, but we do not wish to spend too much energy in the process. Note that the energy of states $\ket{\psi^\star_d}$ in the previous section is equal to $(d-1)\Delta/2$, i.e., it grows linearly with the dimension $d$. This makes one wonder whether such energetic states are actually necessary in order to attain a good approximation to $\M(2)$. Or, in other words, how well can we approximate $\M(2)$ when our battery is infinite dimensional, but its energy is bounded?

Call $\M(\bar{E},2)$ the set of two-level POVMs attainable via devices with a battery of energy smaller or equal than $\bar{E}>0$ for $H_B=\sum_{k=0}k\Delta\proj{k}$. In analogy with the previous section, define $\epsilon_C^\star(\bar{E})$, $\epsilon_Q^\star(\bar{E})$ as $\epsilon_C(\{\rho:\tr(\rho H_B)\leq \bar{E}\})$, $\epsilon_Q(\{\rho:\tr(\rho H_B)\leq \bar{E}\})$, respectively. From the previous section and the observation that $\bra{\psi_d^\star}H_B\ket{\psi_d^\star}=(d-1)/2$, it follows that $\M(\bar{E},2)_{\bar{E}\to \infty}\to\M(2)$. Moreover, 

\be
\epsilon_C^\star(\bar{E})=\epsilon_Q^\star(\bar{E})\leq \frac{1}{2}\left\{1-\cos\left(\frac{\pi}{\left\lfloor \frac{2\bar{E}}{\Delta}\right\rfloor +2}\right)\right\},
\ee

\noindent that is, for large $E$, $\epsilon_Q^\star(\bar{E})\lesssim\frac{0.6168\Delta^2}{\bar{E}^2}$. We wonder whether this last bound is tight, or close to tight. Could it be that there exists an energy threshold $E_0>0$ such that $\M(\bar{E},2)=\M(2)$ for $\bar{E}>E_0$?

Again, by eq. (\ref{def_tau}), the answer to these questions depends on how well we can approximate $1$ by $\sum_{k=0}^\infty |c_{k}c_{k+1}|$, under the constraints $\sum_{k}|c_k|^2=1$, $\sum_{k}|c_k|^2k\Delta\leq \bar{E}$.

What we find in Appendix \ref{fixed_e} is that $\tau(\{\sigma:\tr(\sigma H_B)\leq \bar{E}\})=\varphi(\bar{E}/\Delta)$, with

\be
\varphi(z)=\min_{\lambda}\frac{z+\{\mu:j_{\mu-1,1}=2\lambda\}}{2\lambda},
\label{phi}
\ee

\noindent where $j_{n,1}$ denotes the first zero of the Bessel function of the first kind $J_n$. The numerical minimization of this function is not problematic, since all its extreme points are global minima, see Appendix \ref{fixed_e}. For $z\gg 1$, $\varphi(z)$ behaves as 

\be
\varphi(z)\approx 1-\frac{0.9468}{z^2}.
\label{asym_phi}
\ee

\noindent Consequently,

\be
\epsilon_Q^\star(\bar{E})=\epsilon_C^\star(\bar{E})=\frac{1}{2}\{1-\varphi(\bar{E}/\Delta)\}\approx \frac{0.4734\Delta^2}{\bar{E}^2}, \mbox{ for } \bar{E}\gg \Delta.
\ee

The states $\ket{\Psi^\star_{\bar{E}}}=\sum_{k=0}^\infty c_k\ket{k}$ with average energy $\bar{E}$ maximizing eq. (\ref{def_tau}) will be denoted \emph{power states}. Their coefficients $\{c_k\}$ are given by the recurrence formula:

\be
c_{k+1}=\frac{k+\{\mu:j_{\mu-1,1}=2\lambda^\star\}}{\lambda^\star}c_k-c_{k-1},
\ee

\noindent where $\lambda^\star\in \R^+$ denotes the minimizer of (\ref{phi}). 

It is instructive to compare how useful power states are for high precision measurements as opposed to the usual coherent states, of the form $\ket{\alpha}=e^{-\frac{|\alpha|^2}{2}}\sum_{k=0}^\infty\frac{\alpha^k}{\sqrt{k!}}\ket{k}$, with $\bar{E}=|\alpha|^2\Delta$. For coherent states,

\be
\tau(\{\proj{\alpha}\})=\sum_{k=0}^\infty\frac{|\alpha|^{2k+1}}{\sqrt{k!(k+1)!}}\approx 1-\frac{1}{8|\alpha|^2}=1-\frac{\Delta}{8\bar{E}},
\ee

\noindent where the approximation holds for $|\alpha|\gg 1$ \cite{carruthers}. Now, we can quantify measurement precision by $\log_{10}(1-\tau)$, i.e., by the number of significant decimal places up to which we can approximate $\M(2)$. According to this definition, the last equation means that, for $\bar{E}/\Delta\gg 1$, power states \emph{double} the precision we may reach with coherent states of the same energy. It is hence an interesting question whether current technology allows producing power states in the lab.

\section{Characterization of the sets of accessible measurements}
\label{charact}

The results of the last section show that $\M(\C^d,2),\M(\bar{E},2)\not=\M(2)$, and determine an operational distance between these sets. They leave open, though, the problem of characterizing which two-level POVMs $M\in \M(2)$ can be realized with a measurement device with bounded energy spectrum, or average energy. More generally, given a POVM $M\in M(d)$, we would like to decide if such a POVM can be implemented in a $d$-dimensional target system $S$ with non-trivial energy operator $H_S$ with the aid of measurement devices with energy operator $H_B$.

The dual of this problem is also of interest for the quantum information theory community: given an arbitrary vector of operators $V=(V_0,V_1,...,V_n)$, maximize $\sum_{x=0}^n\tr(M_xV_x)$, for $M\in\M(\C^d)$. This class of problems includes the \emph{energy constrained state discrimination problem}: given a device which randomly prepares a multi-level quantum state $\{\rho_i\}_{i=0}^n$ according to some known probability distribution $\{p_i\}_{i=0}^n$, maximize the probability of correctly guessing which state $\rho_i$ was produced by measuring it with a device possessing a $d$-level battery. 

In Appendix \ref{algorithms} we show how to formulate the characterization of $\M(\C^d,2)$ and linear optimizations over $\M(\C^d,2)$ as semidefinite programs \cite{sdp} involving $O(d)$ $2\times 2$ complex matrices. Due to this small scaling, using standard convex optimization packages like sedumi \cite{sedumi}, we found that a normal desktop can carry optimizations over $\M(\C^d)$ for $d>4000$. We used such programs to prove that the state $\sigma^\star$ with energy distribution (\ref{distrib_star}) is not a universal resource state, in the sense that there exist measurements $M\in \M(\C^d,2)$ which are not achievable with such an ancillary state, thus recovering the conclusions of \cite{mehdi}. As shown in Appendix \ref{algorithms}, the algorithm can be easily adapted to characterize the set $\M(\C^d,d')$ for arbitrary (given) hamiltonians $H_B\in B(\C^d), H_S\in B(\C^{d'})$. The corresponding semidefinite program involves $O(dd')$ $d'\times d'$ complex matrices, and is therefore computationally efficient in both the ancilla and the target dimensions.

The next problem is how to characterize $\M(\bar{E},2)$. Since a full characterization of $\M(\bar{E},2)$ would involve optimizations over infinite dimensional energy distributions $\{p_k:k\geq 0\}$, it is unlikely that we can reduce it to a semidefinite program. The approach followed in Appendix \ref{algorithms} is to define two sequences of inner ($\{M_d^I(\bar{E},2)\}$) and outer ($\{M^O_d(\bar{E},2)\}$) approximations of $\M(\bar{E},2)$, i.e., $\M_d^I(\bar{E},2)\subset \M(\bar{E},2)\subset \M^O_d(\bar{E},2)\subset \M(2)$, with $\lim_{d\to\infty}\M^I_d(\bar{E},2)=\lim_{d\to\infty}\M^O_d(\bar{E},2)=\M(\bar{E},2)$. Each of these approximations can be computed via a semidefinite program involving $O(d)$ $2\times 2$ positive semidefinite matrices. Moreover, the speed of convergence of the scheme is bounded by $O(\bar{E}/d\Delta)$ (see Appendix \ref{algorithms}). Maximizing a linear expression $f(M)$ over $M\in\M(\bar{E},2)$ can thus be accomplished by optimizing over the sets $\M^O_d(\bar{E},2)$ and $\M^I_d(\bar{E},2)$. The first optimization will return an upper bound $\hat{f}_d^O$ on the maximal value $\hat{f}=\max_{M\in\M(\bar{E},2)}f(M)$, while the second optimization will return a feasible POVM $M\in \M(\bar{E},2)$ attaining a sub-optimal value $\hat{f}^I_d$. The difference $\hat{f}_d^O-\hat{f}^I_d$ will quantify the error of the approximation. This algorithm can be adapted as well to describe effective measurements in multi-level quantum systems with devices of energy bounded by $\bar{E}$ if the spectrum of $H_B$ is either finite dimensional or admits a sufficiently simple description.

\section{Why can we make quantum measurements at all?}
\label{paradox}

As we pointed out in section \ref{considerations}, in order to perform non-trivial quantum measurements over system $S$, the hamiltonian of our device's battery $H_B$ must be \emph{exactly} resonant with the hamiltonian of the system under observation $H_S=\Delta \proj{1}$. Such a state of affairs seems very unphysical, since it would allow us to distinguish a hamiltonian $\Delta \proj{1}$ from, say, $H_B=(\Delta+\epsilon) \proj{1}$ for arbitrarily small $\epsilon\not=0$. How can we solve this paradox?

One possibility is that the energy difference is accounted for by extra degrees of freedom with continuous energy operator. Suppose that we are actually trying to perform a non-trivial measurement over a two-level system $S$ with $H_S=\Delta \proj{1}$ by coupling it to a two-level battery with energy operator $H_B=(\Delta+\epsilon)\proj{1}$, with $\Delta\gg\epsilon>0$. If $\ket{\psi_B}=c_0\ket{0}+c_1\ket{1}$ and $\epsilon=0$, we would expect to find a $\tau(\psi_B)=|c_0c_1|$. For $\epsilon>0$, on the contrary, the theory predicts $\tau(\psi_B)=0$. Imagine, though, that there is an extra degree of freedom in the lab, i.e., a third quantum system $C$ with wavefunction $\ket{\psi_{C}}=\int_{0}^\infty \sigma^{1/2}(E)dE\ket{E}$, where $\sigma(E)dE$ denotes the energy density of system $C$, and $\{\ket{E}dE\}$ is an (undercomplete) orthonormal basis of energy states. This continuous degree of freedom could correspond, for example, to the kinetic energy of the clock mediating the measurement interaction. As shown in section \ref{perfect}, the value of $\tau$ for a pure ancillary state $\ket{\phi_B}$ with energy density $f(E)dE$ used to control a two-level system with energy operator $H_S$ is given by eq. (\ref{continuous}).

The energy density $f(E)dE$ of the system $BC$ in the previous case can be seen equal to

\be
f(E)=\{\sigma(E)|c_0|^2+\sigma(E-\Delta-\epsilon)|c_1|^2\}dE,
\ee

\noindent and, consequently, the potential for measuring system $S$ is determined by the quantity

\begin{eqnarray}
\tau(\ket{\psi_B}\otimes\ket{\psi_{C}})=&&\int dE (\sigma(E)|c_0|^2+\sigma(E-\Delta-\epsilon)|c_1|^2)^{1/2}\cdot\nonumber\\
&&\cdot(\sigma(E+\Delta)|c_0|^2+\sigma(E-\epsilon)|c_1|^2)^{1/2}dE.
\label{crossing}
\end{eqnarray}

In the particular case in which $\sigma(E)dE$ is very concentrated around the energy values $[E_0-\sigma,E_0+\sigma]$, with $\epsilon\ll\sigma\ll \Delta$, the above expression is very close to $|c_0c_1|$, as expected. Note, though, that due to the localization condition, $\tau(\{\psi_{C}\})\approx 0$, i.e., the auxiliary system $C$ cannot do much by itself.

As an example, imagine that $\ket{\psi_B},\rho_S$ correspond to the quantum description of the inner degrees of freedom of two systems, say, molecules, which are able to move in space. Suppose further that such systems are, at the beginning of the experiment, in a very-well localized state in momentum space (perhaps because they have been sent to collide and thus interact). Call $\ket{\psi}_{C}$, $\ket{\psi}_{D}$ the corresponding wavefunctions, and assume, for simplicity that they are identical and equal to

\be
\ket{\phi}\propto \int d^3\vec{p}e^{-\frac{p^2}{4\sigma}}\ket{\vec{p}}.
\ee

\noindent Under the assumption that the masses of each system are both equal to $m$, we have that

\be
\ket{\psi_{C}}\otimes \ket{\psi_{D}}\propto \int d\vec{P}_T d\vec{p} e^{-\frac{\vec{P}_T^2}{8\sigma}}e^{-\frac{\vec{p}^2}{2\sigma}}\ket{\vec{P}_T}\ket{\vec{p}},
\ee

\noindent where $\vec{p}_C=\frac{1}{2}\vec{P}_T+\vec{p}$, $\vec{p}_D=\frac{1}{2}\vec{P}_T-\vec{p}$. If we further assume conservation of total linear momentum, we can imagine a prior measurement of $\vec{P}_T$ followed by an interaction between the systems $A,S$ and the remaining of $CD$. Now, for whatever outcome of $\vec{P}_T$, the energy density of the remaining of system $CD$ (we only count the kinetic terms) is given by the distribution

\be
\sigma(E)dE= \left(\frac{m}{\sigma}\right)^{3/2}\frac{2}{\sqrt{\pi}}e^{-\frac{mE}{\sigma}}E^{1/2}dE.
\ee

\noindent The application of formula (\ref{crossing}) gives us

\begin{eqnarray}
\tau(ABCD)&&> e^{-\frac{m\epsilon}{2\sigma}}\int_{\epsilon} dE \left(\frac{m}{\sigma}\right)^{3/2}\frac{2}{\sqrt{\pi}} |c_0||c_1|e^{-\frac{mE}{\sigma}}E^{1/4}(E-\epsilon)^{1/4}\nonumber\\
&&> |c_0||c_1|e^{-\frac{3m\epsilon}{2\sigma}}.
\end{eqnarray}

\noindent The first inequality is close to being saturated for $\frac{\sigma}{m}\ll \Delta$, the second one is just an approximation. In either case, one can see that, as $\epsilon\to 0$, we recover the exact result in a continuous way.

\begin{appendix}

\section{General remarks about energy conservation}
\label{general}

\subsection{State evolution under energy conservation}

Let $S$, be a quantum system, with associated Hilbert space $\H_S$. Call $H\in B(\H)$ the energy operator of system $S$, and let $\Omega:B(\H_S)\to B(\H_S)$ be a trace-preserving quantum channel. The following statements are equivalent:

\begin{enumerate}

\item
$\Omega(\rho)$ has the same energy distribution as $\rho$ for all states $\rho\in B(\H_S)$.
\item
$\Omega(\rho)=\sum_{i} A_i\rho A_i^\dagger$, with 

\be
[A_i,H]=[A_i^\dagger,H]=0,\sum_{i}A^\dagger_iA_i=\id_S.
\label{evol}
\ee

\item
There exists a quantum system $A$, with associated Hilbert space $\H_A$, and a unitary operator $U\in B(\H_S\otimes \H_A)$ with $[U,H\otimes \id_A]=0$ such that $\Omega(\rho)=\tr_A\{U(\rho\otimes \proj{0}_A)U^\dagger\}$.
\end{enumerate}

\begin{obs}
Note that, for unitary operators $U\in B(\H_S\otimes \H_A)$, the condition 

\be
\tr\{U\rho_{SA} U^\dagger (H\otimes \id_A)\}=\tr\{\rho_{SA} (H\otimes \id_A)\}
\ee
\noindent for all joint states $\rho_{SA}$ is equivalent to $[U,H\otimes \id_A]=0$.
\end{obs}

\begin{proof}
($1\Rightarrow 2$) Let $H=\sum_n E_n P_n$ be the spectral decomposition of $H$, where $\{P_n\}_n$ denote orthogonal projectors. Let us assume that (1) holds, and, consequently, $\tr\{\rho P_n\}=\tr\{\Omega(\rho) P_n\}$, for all states $\rho$. Since $\Omega$ is a quantum channel, there exist Kraus operators $\{A_i\}_i$ such that $\Omega(\rho)=\sum_i A_i\rho A_i^\dagger$, with $\sum_i A^\dagger_i A_i=\id_S$. Condition (1) thus implies that 

\be
\sum_i A_i^\dagger P_n A_i=P_n,
\ee

\noindent for all $n$. For $n\not=m$, it follows that $\sum_i\tr(P_m A_i^\dagger P_n A_i)=0$, and so 

\be
\tr(P_m A_i^\dagger P_nP_n A_i P_m)=\tr(P_m A_i^\dagger P_n A_i)=0
\ee

\noindent for all $i,m\not=n$. This last condition implies that $P_nA_iP_m=0$ for all $n\not=m$. In other words, $A_i=\left(\sum_{n}P_n\right)A_i\left(\sum_{m}P_m\right)=\sum_n P_nA_iP_n$, and thus $[A_i,H]=[A^\dagger_i,H]=0$ for all $i$.

($2\Rightarrow 1$) If $[A_i,H]=[A_i^\dagger,H]=0$, then $A_i=\oplus_n A^{(n)}_i$, with $A^{(n)}_i=P_nA_iP_n$. Call $\H^{(n)}_S$ the support of $P_n$, and consider the operator $U=\oplus_n U^{(n)}$, with $U^{(n)}$ being unitary operators defined in the subspace $\H^{(n)}_S\otimes\H_A$ as

\be
U^{(n)}\ket{\psi_n}_S\ket{0}_A=\sum_i A_i^{(n)}\ket{\psi_n}_S\ket{i}_A,
\label{build_uni}
\ee

\noindent for all $\ket{\psi_n}\in \H^{(n)}$. Obviously, $[U,H\otimes \id_A]=0$. Moreover, it can be verified that $\sum_{i}A_i\rho_SA_i^\dagger=\tr_A\{U(\rho_S\otimes \proj{0}_A)U^\dagger\}$.

($3\Rightarrow 1$) Since $U$ commutes with $H\otimes \id_A$, $\tr\{U\rho_{SA}U^\dagger (P_n\otimes \id_A)\}=\tr\{\rho_{SA}(P_n\otimes \id_A)\}$, for any initial state $\rho_{SA}\in B(\H_S\otimes \H_S)$. In the particular case $\rho_{SA}=\rho_S\otimes \proj{0}_A$ the previous relation implies that $\Omega(\rho_S)=\tr_A\{U(\rho_{S}\otimes\proj{0}_A)\}$ has the same energy distribution as $\rho_S$.

\end{proof}

\subsection{Measurements under energy conservation}

As before, consider two quantum systems: system $S$, our target (with energy operator $H\in B(\H_S)$) and system $P$, the pointer (with trivial energy operator). We will model the measurement process of the state $\rho_S$ in system $S$ via coupling it with system $P$, initially in state $\proj{0}$, through an energy-conserving quantum channel $\Omega_{SP}$, followed by a von Neumann measurement of the pointer. Let $\{A_i\}_i$ be the Kraus operators of $\Omega$, with $[A_i,H]=[A_i^\dagger,H]=0$. After a fine-grained von Neumann measurement $\{\proj{k}\}_{k=0}^{n-1}$ of system $P$ with outcome $x$, the (unnormalized) state of system $S$ will be 

\be
\sum_i \bra{x}A_i\ket{0}\rho_S\bra{0}A^\dagger_i\ket{x}.
\ee

\noindent The probability of obtaining outcome $x$ will thus be given by $p(x)=\tr(\rho_S M_x)$, with $M_x\equiv \sum_i \bra{0}A^\dagger_i\ket{x}\bra{x}A_i\ket{0}$ (note that $[M_x,H]=0$). The statistics of any measurement under energy conservation are thus given by $p(x)=\tr(\rho_S M_x)$, with

\be
M_x\geq 0, \sum_x M_x=\id_S, [M_x,H]=0.
\ee

Conversely, any measurement associated to a set of operators $\{M_x\}_{x=0}^n$ satisfying the above conditions can be realized by making system $S$ interact with a pointer $P$ via an energy-conserving quantum channel, followed by a von Neumann measurement of the pointer. To see why, define $A_x\equiv\sqrt{M_x}\otimes V^x$, where $V$ is the displacement unitary operator, i.e., $V\ket{k}=\ket{k+1\mbox{ (mod n)}}$. Clearly, $[A_x,H\otimes \id_P]=0$ and $\sum_x A^\dagger_xA_x=\id_{SP}$, so $\Omega(\rho_{SA})=\sum_{x=0}^{n-1}A_x\rho_{SA}A_x^\dagger$ is an energy-conserving quantum channel. Finally, note that $\tr\{\Omega(\rho_S\otimes \proj{0}_P) (\id_S\otimes\proj{x})\}=\tr\{\rho_S M_x\}$.

\section{Algorithms}
\label{algorithms}

In this Appendix we will devise a collection of algorithms to characterize the effective POVMs generated by measurement devices with different constraints in their batteries's energy distribution.

Note that, for a fixed battery state $\sigma_B$ (not just pure) eq. (\ref{many_chains_mix}) can be used straightforwardly to formulate a semidefinite program \cite{sdp} to characterize the set of effective POVMs generatable via $\sigma_B$. However, if we allow $\sigma_B$ to vary over a continuum of possible states $\B\subset B(\H_B)$, eq. (\ref{many_chains_mix}) is no longer an SDP. If we want to optimize over non-trivial sets of battery states, it is thus necessary to find a more appropriate reformulation of the problem.

Given an energy distribution $\{p_{j,k}\geq 0:\sum_{j}\sum_{k=0}^{L(j)-1}p_{j,k}=1\}$, consider first the problem of characterizing all one-qubit POVM elements which can generated by any possible battery with such an energy distribution. 

First of all, we can assume the battery state to be pure. Indeed, suppose that $\sigma_B$ is mixed, and consider a purification $\ket{\psi}_{BB'}$ in an extended space $\H_B\otimes \H_{B'}$. Then, we can write $\ket{\psi}_B$ as

\be
\ket{\psi}_B=\sum_{j}\sum_{k=0}^{L(j)-1}\sqrt{p_{j,k}}\ket{\varphi_{j,k}},
\ee 

\noindent with $\ket{\varphi_{j,k}}$ being eigenvectors of the energy operator $H_B\otimes\id_{B'}$. By redefining the physical POVMs as $M^{j,k}_x\to (\proj{\varphi_{j,k}}+\proj{\varphi_{j,k-1}})\M^{j,k}_x\otimes \id_{B'} (\proj{\varphi_{j,k}}+\proj{\varphi_{j,k-1}})$, we end up with a pure state $\ket{\psi}_B$, with the same energy distribution as $\sigma_B$, defined over a battery system with a non-degenerate hamiltonian $H'_B$ with the same spectrum as $H_B$, and which allows to reproduce the same effective POVM $\{\tilde{M}_x\}_x$.

Now, the effective measurements generated by apparatuses with battery state $\ket{\psi_B}=\sum_{j,k}c_{j,k}\ket{j,k}$ are given by

\be
\tilde{M}_x=\sum_{j}\sum_{k=0}^{L(j)-1}c^*_{j,k-a}c_{j,k-b}(M^{j,k}_x)_{ab},
\label{many_chains}
\ee

\noindent where $c_{j,-1}=c_{j,L(j)}=0$.

Note that we can rewrite this equation as

\be
\tilde{M}_x=\sum_j\sum_{k}\tilde{M}^{j,k}_x,
\ee

\noindent with

\be
\tilde{M}^{j,k}_x\equiv\left(\begin{array}{cc}c^*_{j,k}&0\\0&c^*_{j,k-1}\end{array}\right)M^{j,k}_x\left(\begin{array}{cc}c_{j,k}&0\\0&c_{j,k-1}\end{array}\right).
\label{equi}
\ee

What conditions must $\{\tilde{M}_x^{k,j}\}$ satisfy? By definition, they are positive semidefinite. Also, by the completeness relation, they must satisfy

\be
\sum_{x}\tilde{M}_x^{j,k}=\left(\begin{array}{cc}p_{j,k}&0\\0&p_{j,k-1}\end{array}\right),
\label{completeness2}
\ee 

\noindent where $\{p_{j,k}\equiv |c_{j,k}|^2\}$ denotes the energy distribution of the battery. 

Conversely, for any set of non-negative operators $\{\tilde{M}^{j,k}_x\}_x$ satisfying (\ref{completeness2}), one can find a physical POVM $\{M_x^{j,k}\}_x$ such that $\tilde{M}^{j,k}_x$ satisfies eq. (\ref{equi}). A characterization of the set of POVMs attainable with states of energy distribution $\{p_{j,k}\}$ (remember that $p_{j,-1}=p_{j,L_(j)}=0$) is thus given by

\begin{eqnarray}
\tilde{M}_x=&&\sum_j\sum_{k=0}^{L(j)}\tilde{M}_x^{j,k},\nonumber\\
\mbox{s.t. } &&\tilde{M}^{j,k}_x\geq 0,\sum_{x}\tilde{M}^{j,k}_x=\left(\begin{array}{cc}p_{j,k}&0\\0&p_{j,k-1}\end{array}\right),
\label{sdp_gen}
\end{eqnarray}

\noindent which, in finite dimensions, constitutes a semidefinite program (SDP) \cite{sdp}.

In finite dimensions, it is thus immediate to design an algorithm to characterize all effective POVMs generatable through ancillary states with energy distribution defined via the shadow of a spectrahedron, i.e., all collections of numbers $\{p_{j,k}\}$ such that $\sum_{j,k}A_{j,k}p_{j,k}+\sum_{l}B_ly_l\geq 0$ for some matrices $A_{j,k},B_l$ and some extra (free and fixed) variables $\vec{y}$. Simply turn $\{p_{j,k}\}$ into free variables in program eq. (\ref{sdp_gen}) and add the extra constraint $\sum_{j,k}A_{j,k}p_{j,k}+\sum_{l}B_ly_l\geq 0$.

\subsection{Characterizarion of $\M(\C^d,2)$}

As remarked in section \ref{finite_spec}, we just have to consider batteries with energy operator $H_B=\sum_{k=0}^{d-1}k\proj{k}$. The energy distribution of the states of the battery $\{p_k\}_{k=0}^{d-1}$ is just limited by the constraint $\sum_{k=0}^{d-1} p_k=1$. A program to characterize $\M(\C^d,2)$ is thus

\begin{eqnarray}
\tilde{M}_x=&&\sum_{k=0}^d\tilde{M}_x^k,\nonumber\\
\mbox{s.t. } &&\tilde{M}^k_x\geq 0,\sum_{x}\tilde{M}^k_x=\left(\begin{array}{cc}p_k&0\\0&p_{k-1}\end{array}\right),\nonumber\\
&&p_{-1}=p_{d}=0,\sum_{k=0}^{d-1}p_k=1.
\label{charac_MCD}
\end{eqnarray}

\subsection{Characterization of $\M(\bar{E},2)$}

Our aim is to solve the feasibility problem

\begin{eqnarray}
\tilde{M}_x=&&\sum_{k=0}^\infty\tilde{M}_x^k,\nonumber\\
\mbox{s.t. } &&\tilde{M}^k_x\geq 0,\sum_{x}\tilde{M}^k_x=\left(\begin{array}{cc}p_k&0\\0&p_{k-1}\end{array}\right),\nonumber\\
&&p_{-1}=0,\sum_{k=0}^{\infty}p_k=1, \sum_{k=0}^\infty p_k k\Delta\leq \bar{E}.
\label{ideal}
\end{eqnarray}

\noindent Unfortunately, the above problem is not an SDP, since it involves an infinite number of free variables. 

The way we will solve this problem will involve relaxing or restricting the above program via truncations of order $d$ in order to get inner and outer approximations to the set $\M(\bar{E},2)$. We will then show that such approximations converge to $\M(\bar{E},2)$ and derive some bounds on the speed of convergence.

First, consider the obvious restriction $\M^I_d(\bar{E},2)=\M(\bar{E},2)\cap\M(\C^d,2)$. This is equivalent to program (\ref{charac_MCD}) with the additional restriction $\sum_{k=0}^{d-1}p_k k\Delta\leq \bar{E}$. This program thus constitutes an inner approximation of $\M(\bar{E},2)$.

As for a relaxation of (\ref{ideal}), fix a number $d\in \N$ greater than 1, and note that, for any energy distribution $\{p_k\}$ we have that

\be
\sum_{k\geq d}\sum_x M_x^k=\sum_{k\geq d} \left(\begin{array}{cc}p_k&0\\0&p_{k-1}\end{array}\right)= \left(\begin{array}{cc}P_d&0\\0&P_d+p_{d-1}\end{array}\right),
\ee

\noindent with $P_d=1-\sum_{k=0}^{d-1}p_k$. Also, notice that $\sum_{k=0}^{d-1}kp_k+dP_d\leq \sum_{k=0}^{\infty}kp_k\leq \frac{\bar{E}}{\Delta}$. A suitable relaxation of (\ref{ideal}) is thus the program

\begin{eqnarray}
\tilde{M}_x=&&\sum_{k=0}^{d-1}\tilde{M}_x^k+\tilde{M}_x^d,\nonumber\\
\mbox{s.t. } &&\tilde{M}^k_x\geq 0,k=0,...,d\nonumber\\
&&\sum_{x}\tilde{M}^k_x=\left(\begin{array}{cc}p_k&0\\0&p_{k-1}\end{array}\right),\forall k\leq d-1,\nonumber\\
&&\sum_{x}\tilde{M}^d_x=\left(\begin{array}{cc}P_d&0\\0&P_d+p_{d-1}\end{array}\right),\nonumber\\
&&p_{-1}=0, \sum_{k=0}^{d-1}p_k+P_d=1,\sum_{k=0}^{d-1} kp_k+dP_d\leq \frac{\bar{E}}{\Delta}.
\label{outer}
\end{eqnarray}

Let us call $\M^O_d(\bar{E},2)\supset \M(\bar{E},2)$ the outer approximation of $\M(\bar{E},2)$ so generated. An immediate question is how close $\M^O_d(\bar{E},2)$ and $\M^I_d(\bar{E},2)$ actually are. It turns out that, for any POVM $M\in \M^O_d(\bar{E},2)$ it is fairly easy to construct an approximate POVM $N\in \M^I_{d+1}(\bar{E},2)$. Indeed, define the matrices:

\begin{eqnarray}
&\hat{M}_0^{d}=\left(\begin{array}{cc}P_d&0\\0&p_{d-1}\end{array}\right), \tilde{M}^d_x=0\mbox{ for }x\not=0,\nonumber\\
&\hat{M}_0^{d+1}=\left(\begin{array}{cc}0&0\\0&P_d\end{array}\right), \tilde{M}^{d+1}_x=0\mbox{ for }x\not=0.
\end{eqnarray}

Then the POVM given by 

\be
N_x=\sum_{k=0}^{d-1}\tilde{M}^k_x+\hat{M}^{d}_x+\hat{M}^{d+1}_x
\ee

\noindent is trivially in $\M(\bar{E},2)\cap \M(\C^{d+1},2)$.

Given an arbitrary POVM $M\in\M^O_d(\bar{E},2)$, its distance with respect to $\M^{I}_{d+1}(\bar{E},2)$ is bounded by

\begin{eqnarray}
\dist_Q(M,\M^I_{d+1}(\bar{E}))&&=\frac{1}{2}\max_{\rho_{DQ},N\in \M^I_{d+1}(\bar{E})}\sum_x\|\tr_D\{\rho_{DQ}(M_x-N_x)\otimes \id_Q\}\|_1\nonumber\\
&&\leq \frac{1}{2}\max_{\rho_{BQ}}\sum_x\|\tr_D\{\rho_{DQ}(\tilde{M}_x^{d+1}-\hat{M}^{d}_x-\hat{M}^{d+1}_x)\otimes \id_Q\}\|_1\nonumber\\
&&\leq \frac{1}{2}\max_{\rho_{BQ}}\sum_x\tr\{\rho_{DQ}(\tilde{M}_x^{d+1}+\hat{M}^{d}_x+\hat{M}^{d+1}_x)\otimes \id_Q\}\nonumber\\
&&\leq \frac{1}{2}\left\|\left(\begin{array}{cc}2P_d&0\\0&2(P_d+p_{d-1})\end{array}\right)\right\|_\infty\leq \frac{E}{\Delta(d-1)},
\end{eqnarray}

\noindent where the last inequality comes from the fact that $(p_{d-1}+P_d)(d-1)\leq \bar{E}/\Delta$.

Being this a rough estimation of the distance between $\M^O_d(\bar{E},2)$ and $\M^I_{d+1}(\bar{E},2)$, we recommend the reader to run both programs in order to perform linear optimizations over $\M(\bar{E},2)$.

\subsection{Extension to multi-level systems}

The above results can be extended to characterize the set of effective POVMs implementable in a target system of dimension $d'$ and (non-degenerate) energy operator $H_S$ by a measurement device with (non-degenerate) hamiltonian $H_B$. Indeed, given $H_S=\sum_m E_m\proj{n}$ and $H_B=\sum_c\mu_n\proj{n}$, the total energy operator is described by eq. (\ref{total_energy}). Re-expressed in terms of energy subspaces, the operator reads:

\be
H_T=\sum_s \tilde{E}_s P_s,
\ee

\noindent where 

\be
P_s=\sum_{i=0}^{d(s)-1}\proj{m(s,i)}\otimes\proj{n(s,i)}.
\ee

\noindent Here, $d(s)\leq d'$ denotes the rank of the energy subspace $P_s$, and $\{m(s,i)\}_{i=0}^{d(s)-1}$, $\{n(s,i)\}_{i=0}^{d(s)-1}$ are used to denote the eigenvectors $\ket{m(i,s)}\otimes\ket{n(i,s)}$ of $H_T$ with eigenvalue $\tilde{E}_s$. Note that $\braket{m(s,i)}{m(s,j)}=\braket{n(s,i)}{n(s,j)}=0$, for $i\not=j$. Now, let $\sigma_B=\proj{\Psi_B}$, with $\ket{\Psi_B}=\sum_{s}\sum_{i=0}^{d(s)-1}c_{n(s,i)}\ket{m(i,s)}\otimes\ket{n(i,s)}$. Any physical POVM implemented over the joint system $SB$ must be of the form $\{M_x=\oplus_s M_x^s\}_x$, with $\sum_x M^s_x=P_s$. Then one can check that the effective POVM performed on system $S$ is given by

\be
\tilde{M}_s=\sum_s\tilde{M}_x^s,
\ee

\noindent with 

\be
(\tilde{M}_x^s)_{m(s,i),m(s,j)}=c^*_{s,i}c_{n(s,j)}(M^s_x)_{m(s,i),n(s,j)}.
\ee

As before, it follows that $\{\tilde{M}_x^s\}_x$ are simply limited by the positive semidefiniteness condition and the completeness relation

\be
\sum_x\tilde{M}_x^s=\mbox{diag}(p_{n(s,i)}),
\ee

\noindent where $p_{n(s,i)}=|c_{n(s,i)}|^2$ denote the energy occupation numbers of $\rho_B$.

Characterizing the set of all measurements effected on $S$ by measurement devices with energy occupation numbers describable as the shadow of a spectrahedron can then be trivially formulated as a semidefinite program. That includes the case where such numbers are just constrained by summing up to one, i.e., we can easily characterize the set $\M(\C^d,d')$. The case where the spectrum of the battery system has infinite cardinality can be also attacked by hierarchies of SDPs as in the previous section.

\section{Interpretation of $\dist_C$, $\dist_Q$}
\label{distances}

Imagine a quantum device capable of performing a certain demolition measurement $M^{a}\in\M$, thus producing an output $x$. We are further promised that the device is actually measuring either $M^0\in M$ or $M^1\in\M$ with probability $1/2$. Suppose, then, that we input a given quantum state $\rho$ and obtain an outcome $x$. Let $p_\rho^a(x)=\tr\{\rho M^a_x\}$. Clearly, the strategy which maximizes the probability of guessing $a$ is to choose $a=\arg\max\{p_\rho^a(x):a=0,1\}$. The maximum probability of guessing $a$ by classical means is thus

\begin{eqnarray}
P_C &&=\max_{\rho}\frac{1}{2}\sum_x\max\{p_\rho^0(x),p_\rho^1(x)\}\nonumber\\
&&=\frac{1}{2}+\max_{\rho}\frac{1}{4}\sum_x|p_\rho^0(x)-p_\rho^1(x)|=\nonumber\\
&&=\frac{1}{2}\{1+\dist_C(M^0,M^1)\},
\end{eqnarray}

\noindent with $\dist_C$ defined as in eq. (\ref{dist_C}).

Suppose now that, rather than analyzing the inputs of the device, we let it measure subsystem $D$ of an entangled state $\rho_{DQ}\in B(\H_D\otimes \H_Q)$, where $\H_D$ ($\H_Q$) denotes the Hilbert space corresponding to the device (the auxiliary system). After obtaining output $x$ from the machine, we perform a POVM $N^x=\{N^x_{a}:a=0,1\}$ over system $Q$, whose outcome will be our guess on the value of $a$. The maximum probability of success of this scheme is

\begin{eqnarray}
P_Q&&=\frac{1}{2}\max_{\rho,N^x}\sum_x\tr\{\rho (M^0_x\otimes N^x_0+M^1_x\otimes N^x_1)\}=\nonumber\\
&&=\frac{1}{2}+\frac{1}{2}\max_{\rho,N_x}\sum_x\tr\{\rho (M^0_x-M^1_x)\otimes N^x_0\}=\nonumber\\
&&=\frac{1}{2}+\frac{1}{2}\max_{\rho}\sum_x\tr_+\{\rho^0_x-\rho^1_x\},
\label{intermed}
\end{eqnarray}

\noindent with $\rho^a_x=\tr_D(\rho M^a_x\otimes \id_Q)$ and $\tr_+(T)$ denoting the sum of the positive eigenvalues of operator $T$. Now, it is easy to see that, for any self-adjoint operator $T$, $\|T\|_1=2\tr_+(T)-\tr(T)$. Substituting in (\ref{intermed}), we have that

\begin{eqnarray}
P_Q&&=\frac{1}{2}+\frac{1}{4}\sum_x\|\rho^0_x-\rho^1_x\|_1+\frac{1}{4}\sum_x\tr_Q(\rho^0_x-\rho^1_x)=\nonumber\\
&&=\frac{1}{2}\{1+\dist_Q(M^0,M^1)\},
\end{eqnarray}

\noindent where $\dist_Q$ is defined as in (\ref{dist_Q}) and in the last line we made use of the fact that $\sum_x\tr_Q(\rho^a_x)=\tr(\rho \sum_xM^a_x\otimes\id_Q)=1$, for $a=0,1$.

A natural question is whether $\dist_C$ and $\dist_Q$ are actually different. They can be proven equal if $M^0$, $M^1$ happen to have only two outcomes. Indeed, suppose that such is the case. Then we have that

\be
\dist_Q(M^0,M^1)=\max_{\rho_{DQ}}\|\tr_D\{\rho_{DQ}(M^0_0-M^1_0)\}\|_1.
\ee

\noindent Using standard identities of the trace norm \cite{nielsen_chuang}, we have that

\begin{eqnarray}
\|\tr_D\{\rho_{DQ}(M^0_0-M^1_0)\}\|_1&&=\max_{-\id\leq A\leq \id}\tr\{\rho_{DQ}(M^0_0-M^1_0)\otimes A\}\nonumber\\
&&\leq \|M^0_0-M^1_0\|_{\infty}.
\end{eqnarray}

On the other hand, we have that

\be
\dist_C(M^0,M^1)=\max_{\rho}|\tr\{\rho (M^0_0-M^1_0)\}|=\|M^0_0-M^1_0\|_\infty.
\ee

\noindent Comparison of these two relations yields $\dist_Q(M^0,M^1)=\dist_C(M^0,M^1)$.

However, even in dimension 2 we can find pairs of POVMs $M^0,M^1$ with $\dist_Q(M^0,M^1)>\dist_C(M^0,M^1)$. Take, for instance the continuous POVMs defined by 

\be
M^0_\psi=2 \proj{\psi}d\psi, M^1_\psi=\id_2d\psi,
\ee

\noindent where $d\psi$ denotes the invariant measure on the pure states of $\C^2$. Since both POVMs are invariant under rotations, we can choose $\rho=\proj{0}$ in eq. (\ref{dist_C}). Then we have

\be
\dist_C(M^0,M^1)=\int d\psi \left||\braket{0}{\psi}|^2-\frac{1}{2}\right|=\frac{1}{4}.
\ee

On the other hand, let $\rho=\proj{\phi}$ in eq. (\ref{dist_Q}), where $\ket{\phi}=\frac{1}{\sqrt{2}}(\ket{01}-\ket{10})$. Then, we have that

\be
\dist_Q(M^0,M^1)\geq \int d\psi \left\|\frac{\proj{\psi^\perp}}{2}-\frac{\id_2}{4}\right\|_1=\frac{1}{2}.
\ee

\section{The role of $\tau(\B)$}
\label{role_tau}

The goal of this Appendix is to prove the following result.

\begin{theo}
Let $\B$ be a set of states in $B(\H_B)$, and let $H_B\in L(\H_B)$ be a non-degenerate energy operator admitting a decomposition in terms of maximal chains $j$ as

\be
H_B=\sum_{j}\sum_{k=0}^{L(j)-1}(\nu_j+k\Delta)\proj{j,k}.
\ee

\noindent Then,

\be
\epsilon_C(\M_{\B})=\epsilon_Q(\M_{\B})=\frac{1}{2}\{1-\tau(\B)\},
\ee

\noindent where

\be
\tau(\B)=\max_{\rho\in \B}\sum_j\sum_{k=0}^{L(j)-2}|\bra{j,k+1}\rho\ket{j,k}|.
\label{def_tau_theo}
\ee

Moreover, for any general one-qubit POVM $M=\{M_x\}_x$, the state $\sigma^\star\in\B$ maximizing (\ref{def_tau_theo}) allows to generate an effective two-level POVM $\hat{M}=\{\hat{M}_x\}$, with

\be
(\hat{M}_x)_{ab}\{1+\delta_{ab}(1-\tau(\B)\}.
\ee

\end{theo}

\begin{proof}
Let $\sigma\in \B$ and consider the problem of maximizing $2|(M_0)_{10}|$ over all $M\in\M(\{\sigma\})$. From eq. (\ref{many_chains_mix}), the result is

\be
2\max_{M^{j,k}}\left|\sum_{j}\sum_{k=0}^{L(j)}\bra{j,k+1}\sigma\ket{j,k}(M^{j,k}_0)_{10}\right|.
\ee

\noindent It is clear that the physical POVMs $\{M^{jk}\}$ maximizing the last expression must be chosen such that

\be
M^{j,k}_0=\left(\begin{array}{cc}\frac{1}{2}&\frac{e^{i\theta_{jk}}}{2}\\\frac{e^{-i\theta_{jk}}}{2}&\frac{1}{2}\end{array}\right),
\ee

\noindent with $e^{-i\theta_{jk}}\bra{j,k+1}\sigma\ket{j,k}\in \R^+\cup\{0\}$.

\noindent The maximum value of $2|(M_0)_{10}|$ attainable with $\B$ as a resource thus corresponds to $\tau(\B)$, as defined in eq. (\ref{def_tau}).

Now, let $\sigma^\star\in \B$ be the minimizer of eq. (\ref{def_tau}) [or a very good approximation, since the maximum may not be achievable], and  suppose that we wish to approximate the POVM $M\in\M$. Take the physical POVMs to be 

\be
M^{j,k}_x=\left(\begin{array}{cc}1&0\\0&e^{-i\theta_{jk}}\end{array}\right)M_x\left(\begin{array}{cc}1&0\\0&e^{i\theta_{jk}}\end{array}\right), 
\ee

\noindent for all $k=0,...,d$. From eq. (\ref{many_chains_mix}) we have that the final POVM $\hat{M}$ is then given by

\begin{eqnarray}
(\hat{M}_x)_{00}&&=(M_x)_{00}\sum_{j}\sum_{k=0}^{L(j)}\tr(\sigma\proj{j,k})=(M_x)_{00},\nonumber\\
(\hat{M}_x)_{11}&&=(M_x)_{11}\sum_{j}\sum_{k=0}^{L(j)}\tr(\sigma\proj{j,k-1})=(M_x)_{11},\nonumber\\
(\hat{M}_x)_{10}&&=(M_x)_{10}\sum_{j}\sum_{k=0}^{L(j)}e^{i\theta_{jk}}\bra{j,k}\sigma\ket{j,k-1}\nonumber\\
&&=(M_x)_{10}\tau(\B),
\end{eqnarray}

\noindent and, analogously, $(\hat{M}_x)_{01}=(M_x)_{01}\tau(\B)$. The second part of the theorem is thus proven.

Let us now demonstrate the first part. Consider $\bar{M}=\{(\id+(-1)^a\sigma_x)/2\}\in \M$. Then, by the triangle inequality, for any $N\in\M(\B,2)$, we have that

\begin{eqnarray}
\dist_C(\bar{M},N)=&&\frac{1}{2}\max_{\rho}\sum_{x=0,1} |\tr\{\rho(\bar{M}_x-N_x)\}|+\sum_{x\geq 2}|\tr\{\rho N_x\}|\geq\nonumber\\
&&\geq \frac{1}{2}\max_{\rho}\sum_{x=0,1} |\tr\{\rho(\bar{M}_x-N'_x)\}|=\dist_C(\bar{M},N'),
\end{eqnarray}

\noindent with $N'\in \M(\B)$ defined as $N'_0\equiv N_0$, $N'_1\equiv \sum_{x>1}N_x$. 

It follows that, in computing $\dist_C(\bar{M},\M(\B,2))$ we can restrict to two-outcome POVMs $N\in\M(\B,2)$. For any such POVM we have that

\begin{eqnarray}
\dist_C(\bar{M},\M(\B))&&\geq\dist_C(\bar{M},N)=\frac{1}{2}\max_{\rho}\sum_{x=0,1} |\tr\{\rho(\bar{M}_x-N_x)\}|=\nonumber\\
&&=\|\bar{M}_0-N_0\|_\infty\geq \left|\frac{1}{2}-\sup\{(N_0)_{01}:N\in\M(\B,2)\}\right|=\nonumber\\
&&=\frac{1}{2}\{1-\tau(\B)\}.
\label{upper}
\end{eqnarray}

Finally, call $f:\M\to \M$ the map which transforms each POVM $M$ into $\hat{M}$, as defined in the theorem, and note that $\dist_Q(M,N)$ can be written as

\be
\dist_Q(M,N)=\frac{1}{2}\max_{\rho,S^x}\sum_x\tr\{\rho(M_x-N_x)\otimes S^x\},
\ee

\noindent with $-\id\leq S^x\le \id$. Then, for any POVM $M\in\M$, we have that

\begin{eqnarray}
\dist_Q(M,M(\B,2))&&\leq \frac{1}{2}\max_{\rho,S^x}\sum_x \tr\{\rho(M_x-f(M_x))\otimes S^x\} \nonumber\\
&&=\frac{1-\tau(\B)}{2}\max_{\rho,S^x}\sum_x |(M_x)_{01}|\tr(\rho R^x\otimes S^x)\nonumber\\
&&\leq\frac{1-\tau(\B)}{2}\sum_x\sqrt{(M_x)_{00}}\sqrt{(M_x)_{11}}\nonumber\\
&&\leq \frac{1}{2}[1-\tau(\B)],
\label{lower}
\end{eqnarray}

\noindent where $R_x=e^{i\theta_x}\ket{0}\bra{1}+h.c.$, for some $\theta_x\in [0,2\pi)$ and the last inequality follows from the fact that the vectors $v_x=\sqrt{(M_x)_{00}}$, $w_x=\sqrt{(M_x)_{11}}$ are unitary, since $\sum_x (M_x)_{00}=\sum_x (M_x)_{11}=1$.

From eqs. (\ref{upper}) and (\ref{lower}), we therefore have that 

\be
\frac{1}{2}\{1-\tau(\B)\}\leq\epsilon_C(\B)\leq \epsilon_Q(\B)\leq \frac{1}{2}\{1-\tau(\B)\}.
\ee

\noindent The theorem is proven.

\end{proof}

\section{Computation of $\tau(\C^d)$}
\label{eigen_finite}

It is easy to see that the solution of eq. (\ref{tau_CD}) corresponds to the maximum eigenvalue of a $d\times d$ matrix of the form

\be
A_d=\left(\begin{array}{cccccc}0&1/2& & & & \\1/2&0&1/2& & & \\ &1/2&0&1/2& & \\ & &1/2&0&1/2& \\ & & & & \ddots\end{array}\right).
\label{A_d}
\ee

Let $\ket{\psi}=\sum_{k=0}^{d-1}c_k\ket{k}$ be an eigenvector of $A_d$ with eigenvalue $\mu$. Then, the coefficients $\{c_n\}$ must satisfy:

\be
\frac{1}{2}(c_{k+1}+c_{k-1})=\mu c_k,\mbox{ for } k=0,...,d-1,
\label{recurr}
\ee

\noindent with $c_{-1}=c_{d}=0$. Let us try the ansatz $c_k= \sin[(k+1)\alpha]$. Then, condition (\ref{recurr}) translates as

\be
\frac{1}{2}\{\sin[(k+2)\alpha]+\sin[k\alpha)]\}=\cos(\alpha)\sin[(k+1)\alpha]=\mu\sin[(k+1)\alpha],
\label{ansatz}
\ee

\noindent with $c_{d}=\sin[(d+1)\alpha]=0$. This last condition implies that $\alpha=m\pi/(d+1)$, with $m\in \Z$. From eq. (\ref{ansatz}) we thus have that $A_d$ has $d$ eigenvalues given by

\be
\cos\left(\frac{m\pi}{d+1}\right), m=1,...,d.
\label{eigen_f}
\ee

Since $A_d$ cannot have more than $d$ eigenvalues, it follows that its whole spectrum is contained in (\ref{eigen_f}). The maximum eigenvalue is obtained by taking $m=1$, in which case the corresponding (normalized) eigenvector is:

\be
\ket{\psi}^\star_d=\sqrt{\frac{2}{d+1}}\sum_{k}\sin\left(\frac{(k+1)\pi}{d+1}\right)\ket{k}.
\ee

\section{Computation of $\tau(\bar{E})$}
\label{fixed_e}

The aim of this section is to solve the following problem:

\begin{eqnarray}
\max &&\frac{1}{2}\sum_{k=0}^\infty c^*_kc_{k+1}+c^*_{k+1}c_k\nonumber\\
\mbox{s. t. } &&\sum_{k=0}^\infty|c_k|^2=1,\sum_{k=0}^{\infty}k|c_k|^2\leq E,
\label{primal}
\end{eqnarray}

\noindent where, for simplicity, we have defined $E\equiv\bar{E}/\Delta$.

The dual of this problem is:

\begin{eqnarray}
\mu(\lambda)\equiv -\inf &&\bra{\psi}H_\lambda\ket{\psi}\nonumber\\
\mbox{s. t. } &&\braket{\psi}{\psi}=1,
\label{dual}
\end{eqnarray}

\noindent with 

\be
H_\lambda=\overbrace{\sum_{k}k\proj{k}}^{\hat{H}_B}-\lambda\overbrace{\sum_{k}(\ket{k}\bra{k+1}+\ket{k+1}\bra{k})}^{\Lambda}.
\ee

For finite values of $\lambda$, the infimum in eq. (\ref{dual}) is actually a minimum, i.e., problem (\ref{dual}) can be reformulated as an eigenvalue problem (namely, computing the minimum eigenvalue of $H_\lambda$). To see this, note that we can approximate the above minimization over infinite-dimensional vectors $\ket{\psi}=\sum_{k=0}^\infty c_k\ket{k}$ by finite dimensional optimizations over vectors of the form $\ket{\psi'}\in \C^d$. Indeed, let $\ket{\psi}=\sum_{n=0}^\infty c_k\ket{k}$ be a normalized vector with $\bra{\psi}H_\lambda\ket{\psi}<0$ (it is easy to see that $\mu(\lambda)>0$ for all $\lambda>0$). Since $\|\Lambda\|_\infty=2$, we have that

\be
0<-\bra{\psi}H_\lambda\ket{\psi}\leq -\bra{\psi}\hat{H}_B\ket{\psi}+2\lambda.
\ee

\noindent It follows that $\sum_{n}|c_{n}|^2n< 2\lambda$, and so, for any $D\in \N$, 

\be
\sum_{n>D}|c_n|^2< \sum_{n>D}|c_n|^2n< \frac{2\lambda}{D+1}. 
\ee

Under the assumption that $c_n\geq 0$ for all $n$, it is easy to see that

\begin{eqnarray}
&\bra{\psi}\hat{H}_B\ket{\psi}-\frac{2\lambda}{D+1} < \bra{\psi_D}\hat{H}_B\ket{\psi_D}< \frac{\bra{\psi}\hat{H}_B\ket{\psi}}{1-\frac{2\lambda}{D+1}},\nonumber\\
&\bra{\psi}\Lambda\ket{\psi}-\frac{4\lambda}{D} < \bra{\psi_D}\Lambda\ket{\psi_D}< \frac{\bra{\psi}\Lambda\ket{\psi}}{1-\frac{2\lambda}{D+1}},
\end{eqnarray}

\noindent where $\ket{\psi_D}\propto\sum_{n\leq D}c_n\ket{n}\in \C^{D+1}$ is a (normalized) finite dimensional approximation of $\ket{\psi}$. Since we can make $D$ arbitrarily large, it is obvious that problem (\ref{dual}) can be approximated arbitrarily well by its finite-dimensional analog

\begin{eqnarray}
\min &&\bra{\psi}H^d_\lambda\ket{\psi}\nonumber\\
\mbox{s. t. } &&\ket{\psi}\in \C^d,\braket{\psi}{\psi}=1,
\label{dual_fin}
\end{eqnarray}

\noindent with $H^d_\lambda=\sum_{k=0}^{d-1}k\proj{k}-\lambda\sum_{k=0}^{d-2}(\ket{k}\bra{k+1}+\ket{k}\bra{k+1})$.

Call $\ket{\psi^\star_d}$ the minimizer of this latter problem, and consider the sequence of minimizers $(\ket{\psi_d^\star})_d$. Then, there exists a subsequence $(d_1,d_2,...)$ such that $\lim_{d_i\to\infty}\braket{0}{\psi_{d_i}^\star}=c^\star_0$ converges. Likewise we can find a subsequence $(d'_i )_i$ of $(d'_i )_i$ such that $\lim_{d'_i\to\infty}\braket{1}{\psi_{d'_i}^\star}=\braket{1}{\psi^\star}=c^\star_1$ also converges. Iterating this process, we obtain a sequence of values $(c_0^\star,c_1^\star,...)$. Now, as $\braket{\psi_d}{\psi_d}=1$ and the weight of the terms $\sum_{n>D}n|\braket{n}{\psi^\star_d}|^2$ becomes negligible for large $D$ irrespective of $d$, we have that the vector $\ket{\psi^\star}=\sum_{k=0}^\infty c_k^\star\ket{k}$ satisfies $\|\ket{\psi^\star}\|^2=1$ and $\bra{\psi^\star}H_\lambda\ket{\psi^\star}=\mu(\lambda)$.

Let us thus solve (\ref{dual_fin}). Let $\ket{\psi}=\sum_{k=0}^{d-1}c_k\ket{k}$ be an eigenvector of $H^d_\lambda$ with eigenvalue $-\mu$, with $\mu>0$. Then,

\be
c_{k+1}=\frac{k+\mu}{\lambda}c_k-c_{k-1},
\label{recurrence}
\ee

\noindent with $c_{-1}=c_{d}=0$.

Now, for $\lambda>0$, it is clear that the eigenvector with minimum eigenvalue of $H^d_\lambda$ can be chosen such that $c_k\geq 0$. It is also straightforward that it cannot contain intermediate zeros in its component vector, that is, it cannot be of the form

\be
\ket{\psi}=\sum_{k=0}^Kc_k\ket{k}+\sum_{k=K+L}^{d-1}c_k\ket{k},
\ee

\noindent for any $L> 1$. Indeed, if such were the case, the vector

\be
\ket{\psi'}=\sum_{k=0}^Kc_k\ket{k}+\sum_{k=K+L}^{d-1}c_k\ket{k-L+1}
\ee

\noindent would satisfy $\bra{\psi'}H_\lambda^d\ket{\psi'}<\bra{\psi}H_\lambda^d\ket{\psi}$. Finally, by the recurrence condition (\ref{recurrence}), it can also be seen not to be of the form $(c_1,c_2,...,c_K,0,0,...,0)$ for $\lambda>0$. It follows that $c_k>0$ for $k=0,...,d-1$. For $k=0,...,d-1$ we can thus write the $c$'s as $c_k\equiv \prod_{i=0}^kx_i$, for some $x_k>0$. Substituting in eq. (\ref{recurrence}) we have that

\be
x_k=\frac{1}{\frac{k+\mu}{\lambda}-x_{k+1}},
\label{recurrence2}
\ee

\noindent for $k\geq 1$. We also derive the conditions $x_1=\frac{\mu}{\lambda}$, $x_d=0$.

Iterating (\ref{recurrence2}), we have that 

\be
\frac{\mu}{\lambda}=x_1=\frac{1}{\frac{1+\mu}{\lambda}-\frac{1}{\frac{2+\mu}{\lambda}-\frac{1}{\ddots\frac{}{\frac{d-1+\mu}{\lambda}-x_d}}}}.
\ee

\noindent Since $x_d=0$, this implies that the minimum eigenvalue $\mu$ of $H_\lambda^d$ mus satisfy the characteristic equation

\be
\frac{\mu}{\lambda}-\frac{1}{\frac{1+\mu}{\lambda}-\frac{1}{\frac{2+\mu}{\lambda}-\frac{1}{\ddots\frac{ }{\frac{d-1+\mu}{\lambda}}}}}=0.
\label{charac}
\ee

Moreover, as long as any $\mu$ satisfies eq. (\ref{charac}) with each term in the fraction being different from zero, we can construct a set of weights $\{x_k\}_{k=1}^{d-1}$ such that $(1,x_1,x_1x_2,x_1x_2x_3,...)$ is an eigenvector of $H_\lambda^d$ with eigenvalue $\mu$. Suppose now that $\mu$ fulfills eq. (\ref{charac}), but some intermediate subfraction becomes zero, i.e.,

\be
\frac{k+\mu}{\lambda}-\frac{1}{\frac{k+1+\mu}{\lambda}-\frac{1}{\ddots\frac{}{\frac{d-1+\mu}{\lambda}}}}=0,
\label{intermediate}
\ee

\noindent with $d-1\geq k\geq 1$, and such that no intermediate subfraction of (\ref{intermediate}) is zero. We could then create a vector $(1,x_1,x_1x_2,...,\prod_{i=k}^{K-1}x_i)\in \C^{d-k}$ which happens to be an eigenvector of $H_\lambda^{d-k}+k\id$ with eigenvalue $-\mu$. But the minimum eigenvalue of $H_\lambda^{d-k}+k\id$ is strictly greater than that of $H_\lambda^d$. 

\noindent From the above discussion, it follows that the minimum eigenvalue of $H^d_\lambda$ is minus the maximum $\mu$ that satisfies eq. (\ref{charac}).

What happens when we take the limit $d\to\infty$ in expression (\ref{charac})? In \cite{wall}, page 349, we find the following beautiful identity:

\be
\frac{J_{n-1}(z)}{J_n(z)}=\frac{2n}{z}-\frac{\frac{z}{2(n+1)}}{1-\frac{\frac{(z/2)^2}{(n+1)(n+2)}}{1-\frac{\frac{(z/2)^2}{(n+1)(n+2)}}{1-\ddots}}}.
\ee

Multiplying each numerator and denominator by $\frac{2(n+k)}{z}$ accordingly, the latter expression can be shown equal to

\be
\frac{2n}{z}-\frac{1}{\frac{2(n+1)}{z}-\frac{1}{\frac{2(n+2)}{z}-\frac{1}{\frac{2(n+3)}{z}-\frac{1}{\ddots}}}}.
\ee

Identifying $\lambda=z/2$ in (\ref{charac}), we have that the solution of (\ref{dual}) is the maximum $\mu>0$ such that

\be
\frac{J_{\mu-1}(2\lambda)}{J_{\mu}(2\lambda)}=0.
\ee

By the interlacing properties of the zeros of the Bessel function, this last condition is equivalent to $J_{\mu-1}(2\lambda)=0$, and so the solution of problem (\ref{dual}) is given by the solution of the transcendent equation

\be
j_{\mu-1,1}=2\lambda.
\ee

\noindent From now on, we will denote its solution $\mu(\lambda)$.

How does this relate to our initial problem (\ref{primal})? Let $\ket{\psi}$ be such that $\bra{\psi}\hat{H}_B\ket{\psi}=E$. By definition, we have that

\be
H_\lambda+\mu(\lambda)\geq 0,
\ee

\noindent for all $\lambda>0$. 

\noindent Bracketing this expression by $\ket{\psi}$, we conclude that

\be
\sum_{k}|c_kc_{k+1}|\leq \frac{E+\mu(\lambda)}{2\lambda}.
\label{bound}
\ee

\noindent Minimizing with respect to $\lambda$, we thus arrive at

\be
\sum_{k}|c_kc_{k+1}|\leq \varphi(E),
\ee

\noindent with $\varphi(E)$ defined as in (\ref{phi}). To show that this last inequality is tight, denote by $\ket{\psi_\lambda}$ the eigenvector with minimum eigenvalue of $H_\lambda$ and notice that 

\begin{eqnarray}
-\mu(\lambda+\delta\lambda)+\mu(\lambda)&&=\frac{1}{\braket{\psi_{\lambda+\delta\lambda}}{\psi_\lambda}}\left\{\bra{\psi_{\lambda+\delta\lambda}}H_{\lambda+\delta\lambda}\ket{\psi_\lambda}-\bra{\psi_{\lambda+\delta\lambda}}H_\lambda\ket{\psi_\lambda}\right\}\nonumber\\
&&=-\delta\lambda\frac{\bra{\psi_{\lambda+\delta\lambda}}\Lambda\ket{\psi_\lambda}}{\braket{\psi_{\lambda+\delta\lambda}}{\psi_\lambda}}
\end{eqnarray}

\noindent Dividing by $\delta\lambda$ and taking the limit $\delta\lambda\to 0$, we have that

\be
\bra{\psi_\lambda}\Lambda\ket{\psi_\lambda}=\frac{\partial\mu}{\partial \lambda}(\lambda).
\ee

\noindent It follows that

\be
E(\lambda)\equiv\bra{\psi_\lambda}\hat{H}_B\ket{\psi_\lambda}=\lambda\frac{\partial\mu}{\partial \lambda}-\mu(\lambda).
\label{en_ex}
\ee

Now, it is easy to see that $E(0)=0$ and that $E(\lambda)$ is continuous in $\lambda$. On the other hand, we have that, for $\mu\gg 1$,

\be
j_{\mu,1}=\mu+c\mu^{1/3}+O(\mu^{-1/3}),
\ee

\noindent with $c\approx 1.85575$ \cite{abramo}. The identity $j_{\mu-1,1}=2\lambda$ thus implies that $\mu\approx 2\lambda-2^{1/3}c\lambda^{1/3}+O(1)$. Substituting in expression (\ref{en_ex}), we have that $\bra{\psi_\lambda}\hat{H}_B\ket{\psi_\lambda}$ diverges for $\lambda\to\infty$. All this implies that, for any $E>0$, we can find $\lambda>0$ such that $\bra{\psi_\lambda}\hat{H}_B\ket{\psi_\lambda}=E$. Such a state clearly saturates inequality (\ref{bound}). The relation $\tau(\{\rho:\tr(\rho \hat{H}_B)\leq E\})=\varphi(E)$ is thus proven.

\noindent It is worth noticing that any value of $\lambda=\bar{\lambda}$ locally minimizing the right hand side of (\ref{bound}) must satisfy

\be
\frac{\partial}{\partial \bar{\lambda}}\frac{E+\mu(\bar{\lambda})}{2\bar{\lambda}}=0.
\ee

\noindent Or, equivalently, $E=\bar{\lambda}\partial\mu(\bar{\lambda})/\partial\bar{\lambda}-\mu(\bar{\lambda})$. The state $\ket{\psi_{\bar{\lambda}}}$ thus saturates inequality (\ref{bound}) for $\lambda=\bar{\lambda}$, and, consequently, for $\lambda\not=\bar{\lambda}$ the r.h.s. of (\ref{bound}) cannot become smaller. Any extreme point of the latter function is hence a global minimum, and so computing $\phi(E)$ numerically becomes an easy task.

Finally, let us speak about the asymptotic behavior of $\varphi(E)$. Replacing $\mu(\lambda)$ by $2\lambda-2^{1/3}c\lambda^{1/3}$ in the right hand side of eq. (\ref{bound}), and minimizing with respect to $\lambda$, we find that $\lambda^\star\approx \frac{27E^3}{16c^3}$. And, consequently,

\be
\varphi(E)\approx 1-\frac{4c^3}{27}\frac{1}{E^2}.
\ee

\noindent Substituting the value of $c$, we arrive at (\ref{asym_phi}).

\end{appendix}

\end{document}